\documentclass[11pt]{report} 

\pdfoutput=1

\usepackage{fancyhdr}
\usepackage{pdflscape}
\usepackage{multirow}
\usepackage[utf8]{inputenc} 
\usepackage{geometry}
\usepackage{amsmath, amsthm,amssymb,amsfonts,latexsym}
\newtheorem{proposition}{Proposition}  
\newtheorem{claim}{Claim}  
\newtheorem{lemma}{Lemma}
\newtheorem{corollary}{Corollary}
\newtheorem{theorem}{Theorem}
\newtheorem{conjecture}{Conjecture}

\theoremstyle{definition}
\newtheorem{definition}{Definition}
\newtheorem{example}{Example}
\newtheorem{mechanism}{Mechanism}

\usepackage[square, sort, numbers]{natbib}
\usepackage{hyperref}
\usepackage{hypernat}
\usepackage[pdftex]{graphicx}

\newcommand{\doctitle}{Cake Cutting Mechanisms}
\newcommand{\docauthor}{Egor Ianovski}
\newcommand{\docdate}{\today}
\title{\doctitle{}}
\author{\docauthor{}}
\date{\docdate{}}
\hypersetup{pdftitle = {\doctitle{} | \docauthor{}},
            pdfauthor = {\docauthor{}},
            pdfborder = {0 0 0.5}}

\usepackage{titlesec}
\titleformat{\chapter}[hang]{\bfseries \huge}{\thechapter}{2pc}{}

\usepackage[newzealand]{babel}
\addto{\captionsnewzealand}{}

\begin{document}

\begin{titlepage}

\begin{center}

\includegraphics[width=0.15\textwidth]{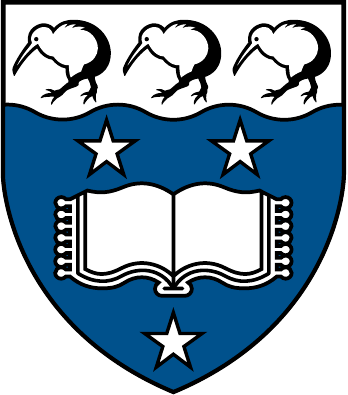}\\[1.5cm]    

\textsc{\LARGE University of Auckland}\\[1cm]

\textsc{\Large Department of Computer Science}\\[0.2cm]

\textsc{\Large Department of Philosophy}\\[0.8cm]

\hrule
\vspace{0.3cm}

{ \Large Cake Cutting Mechanisms}

\vspace{0.3cm}
\hrule
\vspace{0.9cm}
\begin{minipage}{0.4\textwidth}
\begin{flushleft} \large
\emph{Author:}\\
Egor Ianovski
\end{flushleft}
\end{minipage}
\begin{minipage}{0.4\textwidth}
\begin{flushright} \large
\emph{Supervisor:} \\
Dr.~Mark C. Wilson
\end{flushright}
\end{minipage}

\vfill

{\large Submitted in partial fulfilment of the requirements of the degree of
BSc(Hons) in Logic and Computation, 19 October 2011. Last updated}
{\large \today.}

\end{center}

\end{titlepage}

\begin{abstract}
We examine the history of cake cutting mechanisms and discuss the efficiency of
their allocations. In the case of piecewise uniform preferences, we define a
game that in the presence of strategic agents has equilibria that are not
dominated by the allocations of any mechanism. We identify that the equilibria
of this game coincide with the allocations of an existing cake cutting
mechanism. 
\end{abstract}

\setcounter{page}{1}
\renewcommand{\thepage}{\roman{page}}

\tableofcontents

\chapter{Introduction}

\setcounter{page}{1}
\renewcommand{\thepage}{\arabic{page}}

\section{To Cut a Cake}

The topic of cake cutting is a subset of fair division, having its origins in
recreational mathematics. It is the problem of dividing a resource between a
number of agents in a fashion that is ``fair", be it a cake between children or
zoning rights between property developers.

What are the qualities of this resource? Well, a cake is not a bag of sweets.
The resource is continuous, and any given piece can be subdivided into smaller
pieces. A cake is not a mousse. The resource is heterogeneous, different agents
can attach different values to different regions of cake. Finally, a cake is not
meant to be eaten alone. No agent has exclusive right to the cake; it is a
windfall good, its origin unimportant. We have a cake, and we must cut it.

To effect the division of the cake we need more than a kitchen knife. If there
is any hope that the resulting allocation is to have the properties we desire of
it, we need a mechanism: a clearly specified set of rules that incorporates
whatever information it can evoke from the agents to find an allocation that
best satisfies whatever criteria we require of it.

\section{Outline of this Work}

In Section \ref{def} we introduce the mathematical framework in which we will
work throughout the text. Section \ref{lit} we review the existing
cake cutting literature, focusing on mechanisms and their properties. In Section
\ref{eff} we look at the efficiency of allocations, and show that in general
optimal allocations cannot be produced by any cake cutting mechanism.
Part of this problem stems from the strategic behaviour of agents, so in Section
\ref{str} we will consider the equilibria induced by such in
a restricted preferences model of cake cutting. We conclude in Section
\ref{con}.

Appendix \ref{table} summaries all mechanisms mentioned in this text and
Appendix \ref{RW} gives pseudocode presentations of the Robertson-Webb
protocols.

\pagebreak

\subsection{Our Contribution}

Our main result is that the cake cutting mechanism of \cite{Chen10} attains an
allocation that is undominated in terms of utilitarian efficiency. We achieve
this with the help of a game whose equilibrium outcomes are equivalent to
the mechanism's outcomes, which better allows us to isolate the desired property.

\chapter{A Framework for Cake Cutting}\label{def}

Problems in cake cutting have been approached by authors from Mathematics,
Economics, Computer and Political Science. As such terminology is not standard;
different papers use different terms to refer to the same concepts and sometimes
the same terms to different concepts. We will therefore dedicate this section to
introducing notation and definitions as they will be used in this text, which
will allow us to use the same language throughout the literature review in
Section \ref{lit}.

\section{The Cake Cutting Situation}

Central to cake cutting is, of course, the cake. In general we take the cake to
be the unit interval, $[0,1]$, although we shall touch upon the slightly
different context of pie cutting in
Section \ref{litgeo}. Cakes of higher dimensions do arise in the literature, but
that is beyond the scope of the current work.

A cake cutting situation consists
of the
cake and a finite number of agents. If the number of agents is not explicitly
specified, we will reserve $n$ for the number of agents. Every agent has a
utility
function on subsets of cake. We denote agent $i$'s utility function by
$u_i$. We require that $u_i$ be:
\begin{itemize}
\item
Normalised: $u_i([0,1])=1$.
\item
Countably additive: $u_i(X\cup Y)=u_i(X)+u_i(Y)$, where $X$ and $Y$ are disjoint.
\item
Non-atomic: $u_i([a,a])=0$.
\item
Non-negative: $u_i(X)\geq 0$.
\end{itemize}

These requirements are standard. Occasionally (for instance, in \cite{Su99}) an
additional ``hungry agents" condition is required:
\begin{itemize}
\item
Non-zero: $u_i(X)=0$ only if $X$ has zero measure.
\end{itemize}

\pagebreak

Formally, we require that $u_i$ be a probability measure
defined on a $\sigma$-algebra of subsets of the cake. That is:

\begin{align}
u_i([a,b])=\int_a^b\! t_i(x)\, dx
\end{align}
for a probability density function $t_i$. It should be noted that
many papers on the subject do not make this explicit. For most purposes the
exact form of the function is unimportant, it is sufficient that it satisfies
the first four conditions above and that the agents are able to respond to
certain queries regarding their utility; we shall encounter this when we define
Robertson-Webb protocols. Authors that do give definitions tend to give
conflicting ones. \cite{Woodall80}, like us, defines a utility function as given
by
a probability measure, while in \cite{Dubins61} any countably additive real
valued function suffices.

A \emph{slice} of cake refers to a continuous sub-interval of $[0,1]$.
Non-atomicity of utility functions allows us to assume that all slices are
closed. A
\emph{portion} is a union of one or more slices such that any two portions are
disjoint, with the possible exception of boundary points. An $n$-tuple of
portions, $A=(A_1,...,A_n)$, is an \emph{allocation}, with portion $A_i$ being
allocation to agent $i$. Agent $i$ thus derives $u_i(A_i)$ utility from
allocation $A$.

As slices and portions will be of greater interest to us than any other subset
of cake, we will use $|X|$ to refer to the length, rather than the cardinality,
of $X$. That is:
\begin{itemize}
\item
$|[a,b]|=b-a$
\item
$|X\cup Y|=|X|+|Y|$, for disjoint $X$, $Y$.
\end{itemize}

If at any point we wish to refer to the cardinality of set $X$, we will use
$\#X$.

\subsection{Restricted Preferences}

While an arbitrary real valued density function is sufficient for many results
in cake cutting, it leads to problems from the computational side. Almost all
such functions have no finite representation (as there are uncountably many such
functions, but only countably many representations), and accordingly many
associated problems are uncomputable.

One way to circumvent such issues is to restrict the range of admissible
functions. Three such restricted functions, as used in \cite{Chen10} and
\cite{Cohler11}, are given below:

\begin{itemize}
\item
Piecewise uniform: the cake can be partitioned into a finite number of intervals
such that for some constant $c$, $t_i(x)=c$ or 0 over every interval. As
utilities are normalised, $c=1/|P_i|$ where $|P_i|$ is the total length of the
cake that agent $i$ has non-zero density over. For computational purposes, we
require that the endpoints of these intervals be rational numbers.
\item
Piecewise constant: the cake can be partitioned into a finite number of
intervals such that $t_i$ is constant over every interval. In other words, $t_i$
is an arbitrary normalised step function. For computational purposes, we require
that the endpoints of the intervals and the values of $t_i$ be rational numbers.
\item
Piecewise linear: the cake can be partitioned into a finite number of intervals
such that $t_i$ is a linear function over every interval. For computational
purposes, we require that the endpoints of the intervals, slopes and
$y$-intercepts of $t_i$ be
rational numbers.
\end{itemize}

Note that to specify piecewise uniform preferences it is sufficient to specify
which intervals of the cake the agent has non-zero density over. In other words,
the intervals the agent \emph{values}. Due to this representational ease, we
will often use piecewise uniform preferences in examples.

\section{Properties of Allocations}

As the term ``fair division" suggests, an underlying motivation of cake cutting
is the desire to cut the cake in some way that is ``fair". To speak formally of
fairness, we need to define notions of equity. Taking $A=(A_1,...,A_n)$ as an
allocation, we give three of the more prominent definitions here.

\begin{itemize}
\item
Proportionality: $u_i(A_i)\geq 1/n$ for all $i$.
\item
Envy Freeness: $u_i(A_i)\geq u_i(A_j)$ for all $i,j$.
\item
Equitability: $u_i(A_i)=u_j(A_j)$ for all $i,j$.
\end{itemize}

In words, a proportional allocation ensures each of $n$ agents feels that they
received at least $1/n$th of the cake. An envy free allocation ensures that no
agent likes the portion of another agent more than their own. Equitability
ensures all agents derive the same utility from the allocation.

\begin{example}
Given an allocation, we can visualise the equity criteria via an $n\times n$
table. For instance, if $n=3$ we construct:
\begin{center}
\begin{tabular} {l|lll}
&$A_1$&$A_2$&$A_3$\\
\hline
$u_1$&$u_1(A_1)$&$u_1(A_2)$&$u_1(A_3)$\\
$u_2$&$u_2(A_1)$&$u_2(A_2)$&$u_2(A_3)$\\
$u_3$&$u_3(A_1)$&$u_3(A_2)$&$u_3(A_3)$\\
\end{tabular}
\end{center}
The diagonal is precisely the utilities derived from $A$. Hence if all entries
are at least $1/n$, $A$ is proportional. If the
entries in the diagonal are greater than or equal to all other entries in their
row, $A$ is envy free. If all the entries in the diagonal are equal, $A$ is
equitable.
\end{example}

The reader may note that in any cake cutting situation a trivial envy free
allocation is
just $E=(\emptyset,...,\emptyset)$. That is, envy can be eliminated by throwing
the cake away. It would however be difficult to defend such
a manner of attaining equity. As such in addition to equity criteria, certain
concepts of efficiency are beneficial.

\begin{itemize}
\item
Non-wastefulness: for every interval $I$, if $u_i(I)=0$ then $I\subseteq A_i$
only if $u_j(I)=0$ for all $j$.
\item
Pareto efficiency: there is no allocation $B=(B_1,...,B_n)$ such that
$u_i(A_i)\leq u_i(B_i)$ for all $i$ and $u_j(A_j)<u_j(B_j)$ for some $j$.
\item
Utilitarian optimality: there is no allocation with a higher \emph{utilitarian
efficiency} than $A$. That is:
\begin{align}
\sum_{i=1}^n u_i(A_i)\geq\sum_{i=1}^n u_i(B_i)
\end{align}
for all allocations $B$.
\end{itemize}

Observe that the equity criteria do not, in general, imply each other. Consider
these three examples:
\begin{example}
Consider three agents with piecewise uniform preferences. Agent one values
$[0,0.1]$. Agent two and three both value $[0.4,1]$. We construct the allocation
$(A_1,A_2,A_3)=([0,0.1],[0.4,0.8],[0.8,1])$. In other words:
\begin{center}
\begin{tabular} {l|lll}
&$A_1$&$A_2$&$A_3$\\
\hline
$u_1$&1&0&0\\
$u_2$&0&2/3&1/3\\
$u_3$&0&2/3&1/3\\
\end{tabular}
\end{center}
This allocation is proportional as all the entries in the diagonal are greater
than or equal to $1/3$. It is not envy free because $u_3(A_2)>u_3(A_3)$: that
is, agent 3 envies agent 2. It is not equitable as the entries in the diagonal
are not equal.
\end{example}
\begin{example}
Consider two agents with piecewise uniform preferences. Agent one values
$[0,0.5]$, agent two values $[0.5,1]$. We construct the allocation
$(A_1,A_2)=(\emptyset,[0.5,1])$.
\begin{center}
\begin{tabular} {l|ll}
&$A_1$&$A_2$\\
\hline
$u_1$&0&0\\
$u_2$&0&1\\
\end{tabular}
\end{center}
The allocation is envy free, as the diagonal entries are the maxima of their
respective rows. It is not proportional as $u_1(A_1)<1/2$. It is not equitable
as $u_1(A_1)\neq u_2(A_2)$.
\end{example}
\begin{example}
Consider two agents with piecewise uniform preferences. Agent one values
$[0,0.6]$, agent two values $[0.4,1]$. We construct the allocation
$(A_1,A_2)=([0.5,1],[0,0.5])$
\begin{center}
\begin{tabular} {l|ll}
&$A_1$&$A_2$\\
\hline
$u_1$&1/6&5/6\\
$u_2$&5/6&1/6\\
\end{tabular}
\end{center}
The allocation is equitable, as $u_1(A_1)=u_2(A_2)$, but it is neither envy free
nor proportional.
\end{example}

On the other hand, the efficiency criteria are of increasing strength.

\begin{claim}\label{Eff Chain}
An allocation is utilitarian optimal only if it is Pareto efficent and an
allocation is Pareto efficient only if it is
non-wasteful.
\end{claim}
\begin{proof}
Suppose an allocation $A$ is utilitarian optimal but not Pareto efficient. Then
there exists an allocation $B$ such that $u_i(A_i)\leq u_i(B_i)$ for all $i$ and
$u_j(A_j)<u_(B_j)$ for some $j$. But then:
\begin{align}
\sum_{i\neq j} u_i(A_i)&\leq\sum_{i\neq j} u_i(B_i)\\
u_j(A_j)&<u_j(B_j)\\
\sum_{i=0}^n u_i(A_i)&<\sum_{i=0}^n u_i(B_i)
\end{align}
which is impossible because $A$ is utilitarian optimal.

Suppose an allocation $A$ is Pareto efficient but wasteful. That means for some
$i,j$ there exists an interval $I\subseteq A_i$ such that $u_i(I)=0$ but
$u_j(I)>0$. But then we can attain a Pareto dominant allocation by giving $I$ to
$j$, keeping everything else unchanged.
\end{proof}

Finally, though in general the equity criteria are independent, that is not the
case if other requirements are imposed. A result that we will often implicitly
invoke is that
if the entire cake is allocated, then envy freeness implies proportionality.

\begin{claim}\label{envpr}
Given an allocation $A$, if:
\begin{align}
\bigcup_{i=1}^n A_i=[0,1]
\end{align} 
then $A$ is envy free only if it is proportional.
\end{claim}
\begin{proof}
Suppose $A$ is not proportional. Then there exists some agent, $i$, such that
$u_i(A_i)<1/n$. Since utilities are additive and normalised, $\sum
u_i(A_j)>(n-1)/n$ for $j\neq i$. The average value of $u_i(A_j)$ is greater than
$1/n$ and as utilities are non-negative, there must exist some $j$ such that
$u_i(A_j)>1/n>u_i(A_i)$. Which is to say, agent $i$ envies this $j$.
\end{proof}
\section{Mechanisms}

To obtain an allocation, we use a \emph{cake cutting mechanism}. A cake cutting
mechanism is a game played by the agents, which effects a resulting allocation.
We give no general definition of a cake cutting mechanism: one would necessarily
be too broad to be useful. Instead we identify three classes of mechanisms and
motivate them separately.

\subsection{Moving Knife Protocols}

Some of the earliest cake cutting mechanisms, such as those proposed in
\cite{Austin82} and \cite{Stromquist80}, consist of one or more knives being
moved continuously along the cake, stopping when some player yells ``cut!".
\cite{Dubins61} say the following
regarding these protocols:

\begin{quote}
``But their solution is more than a mere existence theorem. In fact, it provides
an important practical method for effecting such a division"
\end{quote}

This is a rather curious feature of cake cutting. While more than an existence
theorem, a moving knife protocol is certainly less
than an effective procedure in the algorithmic sense: the continuous movement of
the knife cannot be captured by a finite protocol. Perhaps a close parallel are
the Japanese and Dutch auctions; in theory price is raised or lowered
continuously until a winner is determined, while in any practical application a
discrete step size would have to be used, and the auction would only approximate
the continuous solution.

A moving knife protocol consists of a finite number of rules with clearly
specified rates and directions of movement, and rules for stopping the knives.
An agent may be asked to move a knife based on information from their own
utility function: for instance, in \cite{Austin82} an agent is asked to move two
knives such that the region between them is worth a half of the cake in the
agent's estimation. However, an agent may not move a knife based on another
agent's estimation, as this information is deemed to be private.

\subsection{Robertson-Webb Protocols}

Robertson-Webb protocols, so named after the authors of \cite{Robertson98},
offer a formalisation that covers most finite cake cutting mechanisms. Agents
are treated as oracles, able to respond to the following two queries:

\begin{itemize}
\item
\texttt{eval(a,b)}: The agent evaluates the slice between $a$ and $b$. That is,
agent $i$ returns $u_i([a,b])$.
\item
\texttt{cut(a,x)}: The agent moves a knife from $a$ to the right until they
measure out a slice they value at $x$. That is, agent $i$ returns a $b$ such
that $u_i([a,b])=x$.
\end{itemize}

A Robertson-Webb protocol is thus an algorithmic procedure taking $n$ agent
oracles as input and returning an allocation of the cake. Throughout the text we
will present Robertson-Webb protocols in a high level, natural language
fashion. Pseudocode formulations are included in Appendix \ref{RW}.

The elegance of this formulation is in its ability to circumvent the
difficulties of dealing with real valued functions. Agents can be assumed to be
hypercomputational entities if need be, able to manipulate their own utility
function sufficiently to respond to the two queries allowed by the mechanism.
Whether or not their utilities have finite representations is of no concern to
the mechanism.

\subsection{Revelation Protocols}

There exist finite protocols, for instance in \cite{Cohler11}, \cite{Chen10} and
\cite{Caragiannis11}, which cannot be represented as a Robertson-Webb protocol.
Instead they take the form of routines which take the agents' utility functions
as input (accordingly some finite representation is required). As agents are
required to directly submit their preferences to the mechanism, these bear some
resemblance to the direct revelation mechanisms of implementation theory. We
will thus refer to them as revelation protocols.

Formally, a revelation protocol is a function mapping $(u_1,...,u_n)\mapsto
(A_1,...,A_n)$. That is, it takes an $n$-tuple of utility functions to an
allocation. This function is not necessarily computable; we will in Section
\ref{littruth} see a mechanism for which no computable implementation is known.
Such mechanisms we call \emph{non-constructive} to distinguish them from
mechanisms proper.

It may not be \emph{a priori} obvious that it is not possible to simulate a
revelation
protocol using a Robertson-Webb protocol. We claim that this is indeed the case,
based on the following observation:

\begin{claim}\label{RW wasteful}
There exist cake cutting situations with piecewise uniform preferences where a
Robertson-Webb protocol cannot create a non-wasteful allocation.
\end{claim}
\begin{proof}
Consider a cake cutting situation with two agents where both agents value the
entire cake uniformly. A non-wasteful allocation in this case is any that
allocates the entire cake.

We will show that after a finite number of Robertson-Webb queries there exists a
different cake cutting situation that would generate the same responses to all
the queries, but where some agent, without loss of generality 1, does not value
the entire cake.

Suppose a finite number of \texttt{eval} and \texttt{cut} queries has been made.
In order to construct the different situation we wish to divide the cake into
intervals. These will be determined by the queries made.

Place a mark on the cake at 0 and 1, all $a,b$ for
every \texttt{eval(a,b)} query to agent 1 and all $a,b$ for every
\texttt{cut(a,x)}
query for agent 1 that returns $b$. Our desired intervals are between
consecutive marks thus placed. We refer to them as pieces.

For every piece, $[i,j]$, in our new situation agent 1 will value
$[i,p]\cup[q,j]$ such that $|[i,p]\cup[q,j]|=|[i,j]|/c$ for some $c>1$. We refer
to all $[p,q]$ so defined as holes.

We claim that the \texttt{eval} queries return the same values: every interval
$[a,b]$ evaluated this way will consist of a finite number of pieces. We reduced
the length of valued cake in every piece by a factor of $1/c$, so the length of
valued cake in $[a,b]$ will be reduced by the same factor. As pieces cover the
entire cake, the total length of valued cake is likewise reduced by a factor of
$1/c$. As the utility derived from a slice of cake with piecewise uniform
preferences is just the length of valued cake in the slice divided by the total
length of valued cake, \texttt{eval} must return the same value.

We claim that the \texttt{cut} queries return the same values: every slice
marked by a \texttt{query} consists of a finite number of pieces, and we have
already seen that the utility derived from the pieces is the same in both
situations.

Now suppose agent 1 is allocated $A_1$ by the mechanism. If $A_1$ has a hole in
it, then this allocation is wasteful. If $A_1$ does not have a hole in it, then
we can replace agent 1 with 2 in the above construction. Agent 2 is allocated
$[0,1]\backslash A_1$, which must then have a hole in it, thus creating waste.
\end{proof}

In other words, a Robertson-Webb protocol cannot find the breakpoints between
intervals an agent values and intervals an agent does not. No such problems
occur with a revelation protocol, as agents simply submit these breakpoints to
the mechanism.

On the other hand revelation protocols cannot, in general, be said to be
stronger than Robertson-Webb protocols. If agents' preferences have some finite
representation then indeed we can simulate a Robertson-Webb protocol with a
revelation protocol, but if we do not have this guarantee then we cannot run a
revelation protocol, while a Robertson-Webb protocol functions equally well.

\subsection{Behavioural Assumptions}\label{defbeh}

A mechanism is a game, and games offer players a choice of strategies to
maximise their utility. In the case of moving
knife protocols, how to stop or move the knife. For Robertson-Webb protocols,
whether to respond to the queries sincerely or otherwise. For revelation
protocols, to submit one's actual utility function or some other which may
result in a preferable outcome.

The allocations produced by a mechanism, therefore, must be understood in terms
of the behaviour the mechanism expects from the agents. A weakly truthful envy
free mechanism faced with fully strategic agents may no longer produce envy free
allocations.

We identify three classes of mechanism in the literature.

\begin{itemize}
\item
Na\"{i}ve mechanism: Agents are assumed to be sincere. When we say a na\"{i}ve
mechanism creates an allocation $A$, we mean that $A$ is the outcome if all the
agents follow exactly the rules specified by the mechanism.
\item
Weakly truthful mechanism: Weak truthfulness was the norm in classical cake
cutting mechanisms. The concept is aptly explained in \cite{Steinhaus48}:

\begin{quote}
``It is easy to prove that the methods explained
here secure to every partner at least a part equal in value to the $1/n$th
of the whole. The greed, the ignorance, and the envy of other partners
can not deprive him of the part due to him in his estimation; he has
only to keep to the methods described above. Even a conspiracy of all
other partners with the only aim to wrong him, even against their own
interests, could not damage him."
\end{quote}

A weakly truthful mechanism, therefore, is one that guarantees every agent a
strategy that will secure them either a proportional or an envy free portion,
regardless of the strategies chosen by other agents. Weak truthfulness has
little meaning outside of proportional or envy free mechanisms: equitability and
the efficiency criteria are essentially global, it makes no sense to say an
agent is guaranteed an equitable portion regardless of the portions of others.

The behaviour that is expected by weakly truthful mechanisms, therefore, is one
of extreme risk aversion. Agents will deviate from sincerity if they can do so
without risk, but failing that will stick to the guaranteed proportional/envy
free portion provided by the mechanism.
\item
Truthful mechanism: These are a recent development in cake cutting and offer
strategy-proofness in the conventional sense - sincerity is a weakly dominant
strategy, and an agent can never increase their expected utility by submitting
an
insincere strategy.
\end{itemize}

The first and last category is relatively sparse, both appearing in rather
recent papers. The better part of the mechanisms we will survey are weakly
truthful, so we will take it as given in Section \ref{lit} that if the
behavioural assumptions of a mechanism are not explicitly specified, the
mechanism is weakly truthful.

\chapter{Literature Review}\label{lit}

We present an overview of the historical developments in cake cutting. In many
ways the core area in cake cutting was the development of mechanisms to procure
envy free allocations and we look at the main results of this in Section
\ref{litenv}. Subsequent sections examine selected themes, mainly those
pertinent to efficiency or strategy. To motivate the subject, we look at a
prehistoric fair division protocol and a generalisation of it presented in
\cite{Steinhaus48}.

\section{Origins}\label{litor}

While the modern treatment of cake cutting can be traced from the middle
20$^{\textrm{th}}$ century, problems of fair division predate recorded history.
Humans are social beings, sensitive to issues of equity, and means of ensuring
it have been around as long as we have.
In particular, our first mechanism is certainly too ancient to be attributed
authorship.

\begin{mechanism}[Cut and Choose]\label{Cut and Choose}
Given two agents, agent $1$ cuts the cake into two slices, $X$ and $Y$ such
that $u_1(X)=u_1(Y)$. Agent $2$ gets assigned a slice of
their choice and agent $1$ gets assigned the remaining slice.
\end{mechanism}

\begin{proposition}
Cut and Choose produces an envy free and proportional allocation.
\end{proposition}
\begin{proof}
For envy freeness, observe that agent 1 cannot be envious because
$u_1(X)=u_1(Y)$. Agent 2 cannot be envious because if $u_2(X)>u_2(Y)$ then they
would choose and be assigned $X$, if $u_2(Y)>u_2(X)$ they would choose and be
assigned $Y$.

For proportionality, we invoke claim \ref{envpr}.
\end{proof}

The idea behind Cut and Choose is simple. Agent 1 is not envious because they
effect the allocation in such a manner that they are indifferent between any
permutation of the portions, agent 2 is not envious because they determine which
permutation is allocated.

Letting one agent choose which possible world to be in is a intuitively
appealing way of eliminating envy for that particular agent. This approach has
been known as far back as Hesiod:

\begin{quote}
`Son of Iapetus, most glorious of all lords, good sir, how unfairly you have
divided the portions!'\\
So said Zeus whose wisdom is everlasting, rebuking him. But wily Prometheus
answered him, smiling softly and not forgetting his cunning trick:\\
`Zeus, most glorious and greatest of the eternal gods, take which ever of these
portions your heart within you bids.'\footnote{Theogony, ll. 543-558. Translated
by Hugh G. Evelyn-White.} \\
\end{quote}

As simple as the mechanism is, the strategic implications of Prometheus'
``cunning trick" illustrate well the distinction between the behaviour expected
by truthful and weakly truthful mechanisms.

Cut and Choose is a weakly truthful mechanism. Regardless of the behaviour of
agent 2, agent 1 can secure an envy free outcome for themselves by following the
rules of the mechanism: if $u_1(X)=u_1(Y)$, agent 1 is indifferent between the
possible allocations. Likewise, no matter how agent 1 cuts the cake, agent 2
gets to pick the piece they value most, hence they have no reason to envy the
other agent.

However, truthful behaviour is not a dominant strategy for agent 1. If agent 1
behaves sincerely they derive $1/2$ utility from the allocation, while if
$u_2(X)\neq u_2(Y)$ agent 2 will attain more than $1/2$. That is not to say
agent 1 has the short end of the stick, however. By anticipating the decision of
agent 2, agent 1 can cut the cake so that $u_2(X)=u_2(Y)+\epsilon$ and
$u_1(X)<u_1(Y)$, agent 2 will pick slice $X$ for slightly more than $1/2$
utility, whereas agent 1 gets slice $Y$ which they value more than the $1/2$
they would have received had they acted sincerely. \cite{Kolm02} dedicates a
section to the behaviour of an expected utility maximising agent under different
information assumptions in the Cut and Choose scenario.

On the other hand, if agent 1 is fully strategic then Cut and Choose may fail to
produce an envy free allocation. If agent 1 has imperfect knowledge of agent 2's
utility, then in trying to cut the cake so that $u_2(X)=u_2(Y)+\epsilon$, agent
1 may underestimate $u_2(Y)$ and agent 2 will pick $Y$, leaving agent 1 with $X$
and envy of agent 2's portion. We thus reiterate the remark at the end of
Section \ref{defbeh} that throughout the literature review if the behavioural
assumptions of a mechanism are not specified, it is taken to be weakly truthful.

During the 1940s Steinhaus, Banach and Knaster sought to extend the Cut and
Choose mechanism to an arbitrary number of agents. In \cite{Steinhaus48} they
present a proportional mechanism for $n$ agents:

\begin{mechanism}[Last Diminisher]\label{Last Diminisher}
Given $n$ agents, the first agent cuts a slice $X$ such that $u_1(X)=1/n$. If
there exists an agent $i$ such that $u_i(X)>1/n$, agent $i$ trims $X$ into $X'$
such that $u_i(X')=1/n$. The trimmings are returned to the cake. The process
continues until no such $i$ exists. The trimmed slice is allocated to the last
agent to trim it, and the procedure recurses on the remaining agents and the
remaining cake.
\end{mechanism}

\begin{proposition}
Last Diminisher produces a proportional allocation.
\end{proposition}
\begin{proof}
It is clear that if an agent is allocated a slice, they perceive that slice to
be at least $1/n$ of the cake. It remains to show that the mechanism can always
make such an allocation. That is, after $i$ agents have been allocated, the
remaining cake is perceived to be at least $(n-i)/n$ of the original by the
remaining agents. We proceed by induction.

Base case: without loss of generality, relabel the agents such that 1 be the
first agent to be allocated a slice. Let $X$ be the
slice allocated to 1. We claim that $u_j(X)\leq 1/n$ for all $j$. Assume
otherwise: that is, for some $j$, $u_j(X)>1/n$. Then abiding by the rules of the
protocol, $j$ would have trimmed $X$ to some smaller $X'$ such that
$u_j(X')=1/n$, and $j$ would have been allocated the first slice instead of 1.
As such, by the additivity of utility functions, $u_j([0,1]\backslash
X)\geq(n-1)/n$ for all remaining $j$.

Inductive case: relabel the agents such that $1,...,i$ are the first $i$ agents
to be allocated a slice.
By the inductive hypothesis, $u_j(\mathcal{R})\geq (n-i)/n$ for
$j\notin\{1,...,i\}$ where $\mathcal{R}$ is the remaining cake. Let $i+1$ be the
next agent to be allocated a slice. Call it $X$. Observe that $u_j(X)\leq 1/n$
for all $j\notin\{1,...,i\}$ by the same argument as before. So by additivity,
$u_j(\mathcal{R}\backslash X)\geq (n-i-1)/n$ for all remaining $j$.

We have thus established that for any $i$, the remaining $n-i$ agents view the
remaining cake as at least $(n-i)/n$ of the original cake. As such the
mechanism can always allocate an agent a slice they perceive to be at least
$1/n$ of the cake.
\end{proof}

Unfortunately, this mechanism fails to be envy free. While an agent can never
envy those who have been allocated before them, it is entirely possible for them
to envy some agent that gets allocated a slice later in the protocol.

\begin{example}
Consider three agents with piecewise uniform preferences. Agent 1 values the
entire cake, agent 2 values $[2/5,1]$, agent 3 values $[4/5,1]$. Agent 1
will be allocated $[0,1/3]$ first, then agent 2 $[5/15,9/15]$, and agent 3 the
remaining $[9/15,1]$.
\begin{center}
\begin{tabular} {l|lll}
&$A_1$&$A_2$&$A_3$\\
\hline
$u_1$&$1/3$&$4/15$&$2/5$\\
$u_2$&0&$1/3$&$2/3$\\
$u_3$&0&0&1\\
\end{tabular}
\end{center}
So agents 1 and 2 envy 3. The envy of 2 towards 3 could be eliminated by using
Cut and Choose once only two agents remain, but agent 1 would still envy 3.
\end{example}

As it turns out, the problem of finding envy free allocations is far more
difficult.

\section{Envy Free Protocols}\label{litenv}

To avoid dealing with uninteresting cases, for the duration of this section we
will only consider mechanisms that allocate the entire cake, as this prevents
the empty allocation from being a solution.

An envy free protocol for the three agent case was discovered by Selfridge,
first published in \cite{Woodall80}.

\begin{mechanism}[Selfridge's Algorithm]\label{Selfridge's Algorithm}
Agent $1$ cuts the cake into slices $X,Y,Z$ such that $u_1(X)=u_1(Y)=u_1(Z)$.
Without loss of generality, we can relabel the slices such that $u_2(X)\geq
u_2(Y)\geq u_2(Z)$. Agent
$2$
trims slice $X$ into $X'$ and $T$ such that $u_2(X')=u_2(Y)$. Agent $3$ picks
whichever of $X'$, $Y$ and $Z$ they prefer, agent $1$ picks one of the two
remaining and agent $2$ gets the last slice. It remains to divide $T$.

There are two cases in the division of $T$. If agent 1 chose slice $X'$ then
$T$ is divided between $2$ and $3$ using Cut and Choose.

Otherwise let whichever of $2$ and $3$ chose slice $X'$ be $x$ and the other
$y$.
Agent $y$
cuts $T$ into $U,V,W$ such that $u_y(U)=u_y(V)=u_y(W)$. Agent $x$ picks
whichever
slice they prefer, $1$ picks from the remaining two and $y$ is allocated the
last
slice.
\end{mechanism}

\begin{proposition}[\cite{Woodall80}]
Selfridge's Algorithm produces an envy free allocation.
\end{proposition}

It pays to note that while we have defined slices as distinct from portions,
thus far the two notions have been used interchangeably. Selfridge's Algorithm
is the first we cover where an agents' portion consists of more than one slice.
As it turns out, this is not coincidental. While proportional mechanisms can and
do allocate contiguous intervals to agents, envy free Robertson-Webb protocols
need necessarily fragment the portions.

\begin{theorem}[\cite{Stromquist08}]
A Robertson-Webb protocol cannot produce an envy free allocation for more than
two agents if the agents' portions consist of a single slice each.
\end{theorem}

The above theorem hinges on the nature of such mechanisms, not of the nature of
the cake. In fact, envy free allocations where portions consist of single slices
always exist (\cite{Stromquist80}). This is not the first time we will run into
mathematical existence and algorithmic impossibility: this should not be
surprising as measure theory lives among the Reals, while algorithmics with the
Integers. If we allow the mechanism to be non-algorithmic, there is no
impossibility. In the same paper, a continuous mechanism is
presented to effect such an allocation for three agents:

\begin{mechanism}[Four Knives]\label{Four Swords}
A sword is moved continuously left to right across the cake, dividing it into
left and right slices, $X$ and $Y$. Three agents move knives across $Y$ such
that each agents' knife splits $Y$ into what they consider two even slices,
$Y_1$ on the left and $Y_2$ on the right. Whenever $u_i(X)=1/3$, agent $i$ yells
``cut!". The cake is cut by the sword and the middle knife, splitting it into
$X$, $Y_1$ and $Y_2$. Agent $i$ receives $X$. If the agent whose knife is
nearest to the sword is not $i$, they take $Y_1$. If the agent whose knife is
farthest from the sword is not $i$, they take $Y_2$. If the agent whose knife
cut the cake is not $i$, they take whichever slice is left over.
\end{mechanism}
\begin{proposition}[\cite{Stromquist80}]
Four Knives produces an envy free allocation.
\end{proposition}

Neither of these two mechanisms generalise to larger numbers of agents as Last
Diminisher did. Part of the difficulty lies in the fact that the proportionality
of the portions allocated thus far will not be affected by whatever allocations
the mechanism may make in the future. Once an agent is allocated a portion they
perceive to be worth at least $1/n$, whatever portions the other agents receive
will not alter the fact that the agent's portion is a proportional one. This is
not the case with envy free procedures; envy can rear its head at any stage of
the allocation.

One approach to this difficulty draws on a moving knife procedure of
\cite{Austin82} that allows two agents to find an allocation where both agents
consider either piece to be worth half the cake - what is called a
\emph{perfect} allocation,
a concept to which we will return in Section \ref{littruth}.

\begin{mechanism}[Austin's Scheme]\label{Austin's Scheme}
A knife is moved from the left across the cake, separating it into $X$ and $Y$.
When $u_i(X)=1/2$, agent $i$ yells ``stop". Agent $i$ takes the knife, adds a
new knife to the left edge of the cake and moves the two knives across in such a
manner such that the region between the knives is always $1/2$ of the cake in
$i$'s estimation. When the region outside the knives is worth $1/2$ in the
second agent's estimation, that agent yells ``stop". $i$ gets the slice between
the knives and the other agent gets the rest of the cake.
\end{mechanism}

\begin{proposition}[\cite{Austin82}]
Austin's Scheme produces a perfect allocation.
\end{proposition}

By iterating Austin's Scheme one can cut the cake into $2^m$ slices such that
agents 1 and 2 think all slices are worth the same. This idea is used by
\cite{Brams97} to create an envy free moving knife mechanism for four agents.

\begin{mechanism}[Four Agent Moving Knife]\label{4 Agent Moving Knife}
Agents $1$ and $2$ use Austin's Scheme to cut the cake into $U$ and $V$, then
use Austin's Scheme on $U$ and $V$ to end up with four slices, $X,Y,Z,W$,
such that $1$ and $2$ consider each of the slices to be $1/4$ of the entire
cake. Agent $3$ trims the most valuable slice in their estimation, without loss
of generality $X$,
into $X'$ such that there exists a tie between $X'$ and the second most valuable
slice.

Agent $4$ picks the slice they value most. If agent $4$ did not pick $X'$, agent
$3$ is allocated $X'$. Otherwise, agent $3$ picks the slice they value most.
Agents $1$ and $2$ pick the remaining slices in any order. It remains to divide
the trimmings.

Rename agents 3 and 4 into $x$ and $y$ where $x$ is the agent that picked $X'$.
Agent
$y$ and 2 use Austin's Scheme on the trimmings to divide it into four slices
they consider to be all worth the same, $T_1,$ $T_2,$ $T_3$ and $T_4$. Agent $x$
picks
a slice of their choice, then 1, then $y$, then 2.
\end{mechanism}

\begin{proposition}[\cite{Brams97}]
Four Agent Moving Knife produces an envy free allocation.
\end{proposition}

What allows the mechanism to divide the trimmings without generating envy is
that agents 1 and 2 have an ``irrevocable advantage" over the player that chose
$X'$. Even if that player were to be allocated the entirety of the trimming, 1
and 2 would not envy that player because that would only bring their portion
back
up to $X$, which 1 and 2 value as much as their own.

\cite{Brams95} capitalise on the idea of irrevocable advantage to create an envy
free mechanism for any number of agents. Unfortunately the details of the
mechanism are too complex to give here. The general procedure involves having
one agent cut the cake into $n$ slices they consider equal, and a preliminary
allocation of these slices made. Whenever this creates envy, a subroutine is run
between the envied and the envier until the envier has an irrevocable advantage
over the envied.

This mechanism is guaranteed to produce an envy free allocation in a finite
number of steps, but this number is unbounded: for any $c$ there exists a cake
cutting situation in which the mechanism will run for more than $c$ steps. No
bounded Robertson-Webb protocol for four or more agents is known.

\subsection{Summary}

The mechanisms presented in this section do not represent the entirety of the
envy free cake cutting literature, but they do cover all cases for which a
solution is known.  

\begin{center}
\begin{tabular}{l | l | l}
& Robertson-Webb & Moving knife\\
\hline
2 agents&Cut and Choose & \cite{Austin82}\\
3 agents & Selfridge, presented in \cite{Woodall80} & \cite{Stromquist80}\\
4 agents & \cite{Brams95} (unbounded) & \cite{Brams97}\\
5 or more & \cite{Brams95} (unbounded) & None known\\
\end{tabular}
\end{center}

Note that we did not cover any revelation protocols: these require that the
agents have finitely representable utility functions and the core areas of cake
cutting do not allow that assumption.

To date no bounded protocol, Robertson-Webb or moving knife, is known for five
or more agents. While we did not explicitly state so, the reader can easily
verify that all mechanisms presents before that of \cite{Brams95} do terminate
in a bounded number of ``steps": queries in the case of Robertson-Webb
protocols, cuts in the case of moving knife mechanisms. However the fact that
this fails in \cite{Brams95} suggests that the query and cut complexity of
mechanisms may be interesting in its own right. We examine this in the next
section.

\section{Query Complexity}\label{litcomp}

The standard approach to measuring the complexity of procedures in Computer
Science is to bound the growth of the running time with respect to the input. To
do so in the context of cake cutting, we need a procedure that can run on an
arbitrary number of agents. We have already seen such a procedure in Mechanism
\ref{Last Diminisher}: Last Diminisher. A natural starting point is to inquire
as to the complexity of this mechanism.

\begin{claim}
The query complexity of Last Diminisher is $O(n^2)$.
\end{claim}
\begin{proof}
The reader may find the pseudocode formulation in Appendix \ref{RW} helpful.

There is a nested loop at play here: we have one agent cut a slice of cake, then
all the remaining agents evaluate and possibly trim the slice. For every agent
allocated we thus have at worst $2n$ queries, and as we allocate $n$ agents the
upper bound is $O(n^2)$.
\end{proof}

\cite{Even84} improve on this bound. They present a proportional mechanism
that takes $O(n\log n)$ queries.

\begin{mechanism}\label{Even nlogn}
Have every agent mark the midpoint, rounding down in the case of an odd number of agents, of the cake in their own valuation. That is,
agent $i$ marks $m_i$ such that $u_i([0,m_i])/u_i([m_i,1])=\lfloor n+1\rfloor / \lceil n+1\rceil$. Observe that such
an $m_i$ is just a real number, so we can define $\prec$ as follows: if
$m_i<m_j$ set $m_i\prec m_j$, if $m_i=m_j$ break the tie arbitrarily. Let $m_j$
be the $\lfloor n/2\rfloor$th mark in this order. Recurse on two subroutines:
one with agents $i$ for $m_i\preceq m_j$ and cake $[0,m_j]$ and the other on
agents $i$ for $m_j\prec m_i$ and cake $[m_j,1]$.

If only one agent is left in a subroutine, allocate them all the cake in the
subroutine.
\end{mechanism}

\begin{proposition}[\cite{Even84}]
Mechanism \ref{Even nlogn} takes $O(n\log n)$ queries and produces a
proportional allocation.
\end{proposition}

As it turns out this is the best we can do. While in the same paper
\cite{Even84} present a randomised protocol that takes $O(n)$ cuts on average,
as far as worst case complexity goes a lower bound was proven by
\cite{Edmonds06}.

\begin{proposition}[\cite{Edmonds06}]
The lower bound on the query complexity of proportional Robertson-Webb
mechanisms is $\Omega(n\log n)$.
\end{proposition}

Given Mechanism \ref{Even nlogn}, this bound is clearly tight.

A lower bound for envy free mechanisms was given in \cite{Procaccia09}. However
as thus far no bounded, $n$-agent envy free procedures are known, the actual bound may
well be higher.

\begin{proposition}[\cite{Procaccia09}]
The lower bound on the query complexity of any envy free Robertson-Webb
mechanism
is $\Omega(n^2)$.
\end{proposition}

\section{Cutting Pies}\label{litgeo}

The distinction between cakes and pies, in the eyes of a mathematician, is that
cakes are square and pies are round. A pie is then identified with $[0,1]$ where
0
and 1 are topologically identical. That is to say, $[0.9,0.1]$ is a slice of
pie, but not of cake.

Note that if we allow portions to consist of any number of slices, there is no
difference between the two problems. $[0.9,0.1]$ may not be a slice, but
$[0.9,1]\cup[0,0.1]$ is clearly a portion, and given additive utilities it is
valued the same. In this section, then, we take it that a portion can consist of
only one slice.

Pies are of interest to us primarily because of two impossibility results.

\begin{proposition}[\cite{Stromquist07},\cite{Thomson07}]\label{piefair}
There exist pie cutting situations where no allocation is both envy free and
Pareto efficient.
\end{proposition}
\begin{proposition}[\cite{Thomson07}]\label{pietruth}
Truthful pie cutting mechanisms cannot produce Pareto efficient allocations.
\end{proposition}

To this point we have dealt solely with the equity side of the problem, so it is
interesting that Proposition \ref{piefair} suggests that issues of efficiency
may be more closely intertwined with equity than first apparent. Proposition
\ref{pietruth} hints that there are difficulties involved in inducing truthful
behaviour, which we shall return to in Section \ref{eff}.

\section{The Price of Fairness}\label{litfair}

An important concept in Economics is the tradeoff between equity and efficiency.
Cake cutting is no different, and the efficiency loss imposed by our equity
criteria has been studied in \cite{Caragiannis09} and \cite{Aumann10}. The first
paper is connected with utilitarian efficiency only, the second introduces the
notion of egalitarian efficiency:
\begin{align}
\mathcal{UE}(A)=\sum_{i=1}^n u_i(A_i)
\end{align}
\begin{align}
\mathcal{EE}(A)=\min_{i} u_i(A_i)
\end{align}

The authors define the \emph{price of proportionality} (respectively: envy
freeness and equitability) with respect to utilitarian efficiency (respectively:
egalitarian) to be the ratio of the utilitarian optimum to the proportional
allocation with the highest utilitarian efficiency. The reader will note that
this value will be different in different cake cutting situations. We are
typically interested in the worst case: that is, the highest possible value of
this instance. By picking extreme situations, therefore, this allows one to
place bounds on the price of these criteria. We will present a result of
\cite{Caragiannis09} to demonstrate this procedure.
\begin{proposition}[\cite{Caragiannis09}]
The price of proportionality is at least $\sqrt{n}/2$.
\end{proposition}
\begin{proof}
Consider a cake cutting situation with piecewise uniform preferences, where $n$
is a square. That is, $n=m^2$ for some $m\in\mathbb{N}$. For $i\in\{1,...,m\}$,
$i$ values $[\frac{i-1}{m},\frac{i}{m}]$. All other agents value the entire cake
uniformly. One can verify that the utilitarian optimum would involve allocating
$[\frac{i-1}{m},\frac{i}{m}]$ to $i$, and nothing to $i\notin\{1,...,m\}$. The
utilitarian efficiency of this allocation is $m=\sqrt{n}$.

Next, consider a proportional allocation. In order to maximise utilitarian
efficiency we should minimise the amount of cake we give to agents
$i\notin\{1,...,m\}$. The least we could give each is $1/n$ of the cake, which
would yield $(n-\sqrt{n})\cdot 1/n=\frac{n-\sqrt{n}}{n}$ efficiency. For a large
enough $n$ this is close to 1. This will leave us with $1/m$ of the cake to
divide between the first $m$ agents. No means of doing this can give us more
than 1 efficiency, so the total utilitarian efficiency is at most 2.

The price of proportionality, therefore, is bounded below by $\sqrt{n}/2$.
\end{proof}

We summarise their results in a table:
\begin{center}
\begin{tabular}{l l | l l l}
&Price of:&Proportionality&Envy freeness&Equitability\\
\hline
Countable portions&$\mathcal{UE}$&At least:
$\sqrt{n}/2$&$\sqrt{n}/2$&$(n+1)^2/4n$\\
\cite{Caragiannis09}&&At most: $2\sqrt{n}-1$&$n-1/2$&$n$\\
Single slice portions&$\mathcal{UE}$&At least:
$\sqrt{n}/2$&$\sqrt{n}/2$&$n-1+1/n$\\
\cite{Aumann10}&&At most: $\sqrt{n}/2+1$&$\sqrt{n}/2+1$&$n$\\
&$\mathcal{EE}$&1&$n/2$&1\\
\end{tabular}
\end{center}

\section{Truthful Mechanisms}\label{littruth}

The treatment of fully strategic behaviour in the literature is a recent
development. Aside from the results of \cite{Thomson07} in the context of pies,
there are two papers on the subject, giving us two mechanisms, one of which is
non-constructive. Both of these are revelation protocols, thus require the
additional assumption that agents' preferences have some finite representation.
In the case of Mechanism \ref{procaccia} this is guaranteed by piecewise uniform
preferences, while in Mechanism \ref{nonconstructive} one must bear in mind that
the mechanism may not function on an arbitrary cake cutting situation.

The non-constructive mechanism, discovered independently by \cite{Chen10} and
\cite{Mossel10}, relies on the concept of a \emph{perfect} allocation.

\begin{definition}
An allocation $A$ is \emph{perfect} if $u_i(A_j)=1/n$ for all $i,j$. That is,
every agent thinks every slice is exactly $1/n$ of the cake.

Using our previous means of a table to visualise an equity criteria, a perfect
allocation is where all the entries in the table are $1/n$.
\begin{center}
\begin{tabular} {c|clc}
&$A_1$&$\cdots$&$A_n$\\
\hline
$u_1$&$1/n$&$\cdots$&$1/n$\\
$\vdots$&$\vdots$&$\ddots$&$\vdots$\\
$u_n$&$1/n$&$\cdots$&$1/n$\\
\end{tabular}
\end{center}
\end{definition}

A result of \cite{Alon87} guarantees the existence of perfect allocations.
However this result is purely existential. In fact, such an allocation cannot be
attained by a Robertson-Webb protocol, even for two agents \cite{Robertson97}.
If we could find such an allocation, however, we could use the following
mechanism:

\begin{mechanism}\label{nonconstructive}
Given the agents' preferences, construct a perfect partition,
$(\pi_1,...,\pi_n)$. Randomly assign $\pi_i$ to
some agent. Remove that agent and recurse on the remaining agents.
\end{mechanism}

\begin{proposition}[\cite{Chen10},\cite{Mossel10}]
Mechanism \ref{nonconstructive} is truthful in expectation and produces a
perfect allocation.
\end{proposition}

Such a solution leaves much to be desired. Even the non-constructive nature
aside, the fact that this mechanism is only truthful in expectation means that
it is not robust enough to handle risk seeking agents: a single agent willing to
take a gamble on the outcome could submit an insincere utility function, thereby
the partition constructed by the mechanism would not be perfect at all, and
could well lead to loss of envy freeness and proportionality for the sincere
agents.

Given piecewise uniform preferences, however, a deterministic mechanism which
avoids these difficulties exists.

\begin{mechanism}\label{procaccia}
Let $\mathfrak{A}$ be a subset of agents and $\mathfrak{X}$ a subset of the
cake. Let $D(\mathfrak{A},\mathfrak{X})$ be all
the intervals of $\mathfrak{X}$ that are valued by at least one agent in
$\mathfrak{A}$. Define:
\begin{align}
avg(\mathfrak{A},\mathfrak{X})=\frac{D(\mathfrak{A},\mathfrak{X})}{\#\mathfrak{A
}}
\end{align}
An allocation is said to be \emph{exact} with respect to $\mathfrak{A}$ and
$\mathfrak{X}$ if it
assigns to every agent in $\mathfrak{A}$ a portion of $\mathfrak{X}$ of length
$avg(\mathfrak{A},\mathfrak{X})$ consisting
only of the intervals that the agent values.

Given the set of agents $\mathcal{A}$ and the cake $[0,1]$, find
$\mathfrak{A}\subseteq\mathcal{A}$ such that $\mathfrak{A}$ minimises the value
of
$avg(\mathcal{A},[0,1])$. Produce an exact allocation with respect to
$\mathcal{A}$ and $[0,1]$. Recurse on $\mathcal{A}\backslash \mathfrak{A}$ and
$[0,1]\backslash D(\mathfrak{A},[0,1])$.
\end{mechanism}
\begin{proposition}
Mechanism \ref{procaccia} is truthful and produces an envy free allocation.
\end{proposition}

To date no extensions to more complicated preferences are known.

\chapter{Efficiency of Allocations}\label{eff}

We examine the notions of utilitarian and egalitarian efficiency, asking what it
means for an allocation to be optimal in either of these measures. We place
bounds on their values and discuss the conditions for their existence. In both
cases we demonstrate that such allocations cannot, in general, be produced by
cake cutting mechanisms if the agents are allowed to be strategic.

\section{Utilitarian Efficiency}\label{effut}

We recall the notion of utilitarian efficiency:
\begin{align}
\mathcal{UE}(A)=\sum_{i=1}^n u_i(A_i)
\end{align}

A utilitarian optimal allocation is therefore one which attains the highest
possible utilitarian efficiency. It is easy to put bounds on this value, but due
to the flexibility of a cake cutting situation these aren't very interesting:

\begin{proposition}\label{boundue}
The utilitarian efficiency of a utilitarian optimal allocation is bounded above
by $n$, below by 1, and these bounds are tight.
\end{proposition}
\begin{proof}
For the upper bound, we observe that the maximum utility attained by any one
agent is 1 due to normalisation. A sum of $n$ terms, each bounded by 1, is
bounded by $n$. To see that this bound is tight, consider a cake cutting
situation with piecewise uniform preferences where agent $i$ values
$[\frac{i-1}{n},\frac{i}{n}]$. That is, all preferences are disjoint, so we can
allocate every
agent a portion that they value as much as the entire cake.

For the lower bound, $u_i([0,1])=1$ for any $i$, so we can always give the
entire cake to one agent. To see that this is bound is tight, consider a cake
cutting situation with piecewise uniform preferences where all agents value the
entire cake. That is the utility derived from any portion by any agent is
precisely the portion's length, $u_i(A_i)=|A_i|$. As no allocation can allocate
portions with combined length exceeding that of the cake, the utilitarian
efficiency cannot exceed 1.
\end{proof}

Another simple, yet important, result concerns the existence of utilitarian
optimal allocations.

\begin{theorem}\label{uoexistence}
A necessary and sufficient condition for the existence of a utilitarian optimal
allocation is that it be possible to divide the cake into a finite number of
slices, $S_i$, such that for all $i$, for some $j$, for every sub-interval
$S'_i\subseteq
S_i$, for all $k$, $u_j(S'_i)\geq u_k(S'_i)$.
\end{theorem}
\begin{proof}
Suppose this condition is not satisfied. Let $A$ be any allocation. There must
be some slice, $S$, in some portion, $A_i$, such that for some sub-interval
$S'\subseteq S$, $u_j(S')>u_i(S)$ for some $j$. Then the allocation obtained by
moving $S'$ to $A_j$ and keeping all else equal will have a higher utilitarian
efficiency. As it is always possible to create an allocation with a higher
utilitarian
efficiency, there can be no maximum.

Suppose this condition is satisfied. We claim that the allocation produced by
allocating $S_i$ to the $j$ so defined is optimal. Suppose otherwise. This would
mean that it is possible to improve on this allocation by giving the interval
$I\subseteq [0,1]$ to some other agent. We consider two cases.

Case one: $I\subseteq S_i$ for some $i$. In this case $I$ is already allocated
to a $j$ such that $u_j(I)\geq u_k(I)$ for all $k$. As preferences are additive,
giving $I$ to any other agent cannot increase the utilitarian efficiency.

Case two:$I\subseteq\bigcup S_i$ where $i\in X$ for some $X$. We split $I$ into
$I_i$, such that $I_i=I\cap S_i$. Given additivity of preferences, allocating
each $I_i$ to the agent that values it most will yield at least as much utility
as allocating all of $I$ to some agent. With $I_i$ so defined, we can return to
case one.
\end{proof}

Theorem \ref{uoexistence} may seem to merely restate what a utilitarian optimal
allocation is, rather than provide the conditions for its existence. However
this circumlocution is necessary, and allows us to prove that cake cutting
situations with piecewise linear preferences always admit utilitarian optimal
allocations.

\begin{proposition}\label{plexistence}
Given a cake cutting situation with piecewise linear preferences, a utilitarian
optimal allocation exists.
\end{proposition}
\begin{proof}
We will divide the cake into a finite number of slices, $S_i$, satisfying the
hypotheses of Theorem \ref{uoexistence}.

Recall that with piecewise linear preferences, the cake can be partitioned into
a finite number of intervals such that $t_i$ is linear over every interval. We
will refer to these intervals as pieces. Mark a point on the cake
wherever:

\begin{itemize}
\item
A piece of some $t_i$ begins or ends, or
\item
$t_i(x)=t_j(x)$ for $i\neq j$. That is, wherever the density functions of two
agents intersect.
\end{itemize}

Observe that this constitutes a finite number of marks: we have a finite number
pieces for each of a finite number of agents, and as $t_i$ is linear over every
piece it can only intersect other $t_j$ a finite number of times.

Let $S_i$ then be the slice between the $i$th and $(i+1)$th mark, taking the
mark at 0 to be the first. Observe that there exists a $j$ such that
$t_j|_{S_i}(x)\geq t_k|_{S_i}(x)$ for all $k$. For if not, then either some
$t_k$ must intersect $t_j$ over $S_i$, or there are two pieces of $t_k$ in
$S_i$, such that over one piece $t_j$ is larger, over the other $t_k$. But
neither of these is possible, because we placed marks at every intersection and
every piece endpoint, so over every slice we only have non-intersecting, linear
functions.

As agent $j$'s density therefore is greater over all of $S_i$, it is easy to see
that for every sub-interval
$S'_i\subseteq S_i$, $u_j(S'_i)\geq u_k(S'_i)$.
\end{proof}

\begin{corollary}
Given a cake cutting situation with piecewise constant or piecewise uniform
preferences, a utilitarian optimal allocation exists.
\end{corollary}
\begin{proof}
Either of the two can easily be seen to be a special case of piecewise linear
preferences.
\end{proof}

Piecewise linear preferences are extremely general and can be used to
approximate a wide range of utility functions, so it may well be the case that
for every non-pathological case a utilitarian optimum exists. However, the cake
cutting framework is general enough to admit pathologies where one does not. It
is not terribly difficult to construct such an example using density functions
which oscillate an infinite number of times over the unit interval.

\begin{example}
Consider a cake cutting situation with two agents, their density functions given
by $t_1(x)=\alpha(\sin(\frac{x}{1-x})+1)$ with normalisation constant
$\alpha\approx 0.744391$, and a constant $t_2(x)=1$.

For any partition of the cake into a finite number of slices, we can always
improve on the allocation by splitting a slice on the right of the cake into
two, and allocating each to whichever agent derives more utility from it.
\end{example}

While an infinitely oscillating function is necessary for a counter example, it
is not sufficient. If we replace $t_2$ in the above example with a
piecewise defined:
\begin{align}
t'_2(x)= \left\{
     \begin{array}{lr}
       0 & : x \in [0,0.5)\\
       2 & : x \in [0.5,1]
     \end{array}
   \right.
\end{align}
Then an optimum allocation clearly exists: give $[0,0.5]$ to agent 1, $[0.5,1]$
to agent 2.

The problem arises from the fact that in some situations we can always increase
efficiency by making a finer division of the cake. As such we speculate that the
problem would disappear if the agents' portions were a fixed number of slices.

\begin{conjecture}
A utilitarian optimal allocation always exists in the context where the agents'
portions are restricted to a constant $c$ number of slices.
\end{conjecture}

\subsection{Non-Existence of Mechanisms}

We round off our discussion of the utilitarian optimum by observing that such
allocations are, in general, unattainable.

A utilitarian optimal allocation will generally require a very precise partition
of the cake, and Robertson-Webb protocols cannot obtain enough information about
the agents'
utility functions to do so.

\begin{corollary}\label{RW nonoptimal}
Robertson-Webb protocols cannot always produce utilitarian optimal allocations.
\end{corollary}
\begin{proof}
By Claim \ref{RW wasteful}, Robertson-Webb protocols cannot always produce
non-wasteful
allocations, so by Claim \ref{Eff Chain} they cannot produce utilitarian optimal
allocations.
\end{proof}

More generally, a utilitarian optimum may be against the interests of individual
agents. As such it should be no surprise that mechanisms fail in the face of
strategic agents.

\begin{claim}
There is no cake cutting mechanism that attains a utilitarian optimal allocation
in every cake cutting situation if the agents are strategic.
\end{claim}
\begin{proof}
Consider two situations with three agents with piecewise uniform preferences. In
the first agent 1 values $[0,0.5]$, agent 2 $[0.5,1]$, agent 3 $[0,1]$. In the
second agent 1 values $[0,0.5]$, agent 2 $[0.5,1]$, agent 3 $[0.4,0.6]$.

Observe that in the first situation the unique utilitarian optimal allocation is\\
$A_1=([0,0.5],[0.5,1],\emptyset)$ while in the second
$A_2=([0,0.4],[0.6,1],[0.4,0.6])$.

Suppose a mechanism, given the second situation, produces $A_2$. When faced with
the first situation, agents 1 and 2 have the same preferences as before, and as
such will respond to the mechanism in the same manner. Agent 3 has different
preferences, but that information is not available to the mechanism. As agent 3
derives more utility from $A_2$ than $A_1$, they can pretend to value
$[0.4,0.6]$ instead of $[0,1]$ and the mechanism would be unable to distinguish
between the two situations and would produce a suboptimal allocation in one of
the cases.
\end{proof}

We can also show that there exists no mechanism that always attains a greater or
equal utilitarian efficiency than any other mechanism. This suggests that a
better candidate for the ``best possible" mechanism may be one that is never
dominated, rather than one that always dominates - we return to this in
Section \ref{strchar}.

\begin{claim}
There is no mechanism that in every cake cutting situation produces an
allocation with utilitarian efficiency greater or equal to that produced by any
other mechanism, if the agents are strategic.
\end{claim}
\begin{proof}
Take an arbitrary allocation, $A$, and consider the mechanism that always
allocated $A$, regardless of the situation it is in. If $A$ is  non-empty, we
can
always find a cake cutting situation in which $A$ is actually utilitarian
optimal, so the mechanism will produce a utilitarian optimal allocation in at
least one situation.

Clearly we can define such a mechanism for every possible allocation. If there
existed a mechanism that did at least as well as all of these, it would
necessarily produce a utilitarian optimal allocation in any cake cutting
situation, but if the agents are strategic this is impossible.
\end{proof}

\section{Egalitarian Efficiency}

Recall that egalitarian efficiency was defined in \cite{Aumann10} as:
\begin{align}
\mathcal{EE}(A)=\min_{i} u_i(A_i)
\end{align}

As with utilitarian efficiency, we can define an egalitarian optimal allocation
as one which maximises this value and likewise prove bounds on it.

\begin{proposition}
The egalitarian efficiency of an egalitarian optimal allocation is bounded above
by $1$, below by $1/n$, and these bounds are tight.
\end{proposition}
\begin{proof}
It is easy to see that for any allocation $A$,
\begin{align}
n\cdot\mathcal{EE}(A)\leq\mathcal{UE}(A)
\end{align}

In particular, if $A$ is utilitarian optimal then for any $B$,
\begin{align}\label{eeue}
n\cdot\mathcal{EE}(B)\leq\mathcal{UE}(A)
\end{align}
as otherwise the utilitarian efficiency of $B$ would have been higher than of
$A$.

With Proposition \ref{boundue}, this immediately gives us the upper bound. To
see that it is tight, consider again the case of pairwise disjoint piecewise
uniform preferences.

For the lower bound, we can create an allocation with egalitarian efficiency of
at least $1/n$ by running Last Diminisher or any other proportional mechanism.
To see that it is tight, we invoke \eqref{eeue} and Proposition \ref{boundue}.
\end{proof}

In both situations used in the proof, we in fact had a stronger relation than
that of \eqref{eeue}. The egalitarian efficiency of the egalitarian optimum was
equal to $1/n$ of the utilitarian efficiency of the utilitarian optimum. One
may ask if this is always the case. The answer is no.

\begin{example}
Consider a cake cutting situation with piecewise uniform preferences with three
agents where
agents 1 values $[0,0.5]$, agent 2 values $[0.5,1]$ and agent 3 values the cake
uniformly. The unique utilitarian optimal allocation is
$([0,0.5],[0.5,1],\emptyset)$, with utilitarian efficiency of 2. The egalitarian
optimum, however, is $([0,0.25],[0.75,1],[0.25,0.75])$ with egalitarian
efficiency of only $1/2$.
\end{example}

\subsection{Non-Existence of Mechanisms}

Observe that to attain an egalitarian efficiency higher than $1/n$ is to give
every agent a portion they value at more than $1/n$ of the cake. This coincides
with an equity criterion examined in \cite{Dubins61}.

\pagebreak

\begin{definition}
An allocation $A$ is \emph{super proportional} if $u_i(A_i)>1/n$ for all $i$.
\end{definition}

\cite{Dubins61} prove the existence of super proportional allocations, provided
at least two agents have different utility functions. However \cite{Mossel10}
present an impossibility result for attaining such allocations.

\begin{proposition}[\cite{Mossel10}]
There is no mechanism that produces a super proportional allocation in every
cake cutting situation if the agents are strategic.
\end{proposition}

\begin{corollary}
There is no mechanism that produces an egalitarian optimal allocation in every
cake cutting situation if the agents are strategic.
\end{corollary}
\begin{proof}
If a situation has two agents with different utility functions, a super
proportional allocation exists. Any egalitarian optimum, therefore, must be
super proportional.
\end{proof}

\chapter{Strategic Cake Cutting}\label{str}

Motivated by the existence result of Proposition \ref{plexistence}, we ask
whether we can
construct mechanisms to find utilitarian optimal allocations in the simplest
case - that of piecewise uniform preferences. It turns out that as far as
na\"{i}ve mechanisms go the problem is trivial, which leads us to consider
strategic behaviour.

\section{Two Non-Wasteful Mechanisms}

As we saw in Claim \ref{Eff Chain}, a necessary condition for an allocation
being utilitarian optimal is it being non-wasteful. As such given Claim \ref{RW
wasteful} we can exclude Robertson-Webb protocols from consideration, and in the
interests of keeping our mechanisms conventionally computable we will also
exclude moving knife protocols. This leaves us with revelation protocols, which
suits us well as given our definition of piecewise uniform preferences we are
guaranteed to have a finite representation: an agent need only submit the end
points of every interval they value.

Since for the duration of this section we restrict our attention to piecewise
uniform preferences, we no longer need the generality of our previous definition
of a utility function. We will denote by $P_i$ the preferences of agent $i$:
that is, the union of all the slices agent $i$ values. A cake cutting situation
with piecewise uniform then is completely specified by the $n$-tuple
$(P_1,...,P_n)$. A mechanism will be a
function $(S_1,...,S_n)\mapsto (A_1,...,A_n)$, where $S_i$ is the
\emph{strategy} of agent $i$: a union of slices, not necessarily equal to $P_i$.
The utility an agent attains from an allocation can be seen to be just:
\begin{align}\label{PU utility}
u_i(A_i)=\frac{|A_i\cap P_i|}{|P_i|}
\end{align}

Non-wastefulness is not an onerous condition, but it is sufficient for the most
obvious of trivial solutions - give all the cake to one agent - to fail. If the
agent does not value the entire cake, there may be waste where they receive some
slice from which they derive no utility, but some other agent would have. We
need to be a bit more sophisticated, but not by much.

\pagebreak

\begin{mechanism}[Lex Order]\label{Lex Order}
Form a linear order, $\prec$, over the agents. Allocate every agent $i$:
\begin{align}
A_i=S_i\backslash\bigcup_{j\prec i}S_j
\end{align}
\end{mechanism}

In other words, Lex Order simply gives every agent what they asked for, minus
what was already given to agents coming earlier in the order. It's easy to see
that Lex Order is a truthful mechanism: agents' payoffs are determined solely by
the order, which is exogenous. Submitting some $S_i$ where $|P_i\backslash
S_i|\neq 0$ would certainly not help agent $i$: it will only reduce the chance
of them getting some of the cake that they value. Submitting some $S_i$ where
$|S_i\backslash P_i|\neq 0$ likewise cannot increase their utility. However in
this case it cannot decrease it either. Being truthful in this respect is only
weakly dominant. Unfortunately pretending to value some of the cake that one
does not
can potentially harm the welfare of other agents. If the first agent in the
order claims to value $[0,1]$ while they only value $[0,0.1]$ Lex order may no
longer be non-wasteful. This leads us to define a behavioural restriction.

\begin{definition}
Agents are said to be \emph{well behaved} if $S_i\subseteq P_i$ for all $i$.
\end{definition}

We will generally assume agents are well behaved. This involves the implicit
assumption that ceteris paribus, agents have a bias in favour of truthfulness
and derive no misanthropic pleasure from causing harm to others. If need be,
this behaviour could be enforced by, for instance, imposing an $\epsilon$ cost
on the length of an agent's strategy. If $\epsilon$ is small enough it should
not deter the agent from choosing a strongly dominant strategy if one exists,
but faced with multiple equivalent strategies will choose the smallest -
which would be in accord with well behavedness.

We observe that if one had some prior knowledge of the agents' preferences, one
could easily construct $\prec$ such that Lex Order would produce a utilitarian
optimal allocation: simply order the agents by the length of their preferences.
The question, then, is how one should behave in the absence of such information.
The obvious approach would be to ask the agents, and that is what we will
consider.

\section{Length Game and its Equilibria}

We modify Lex Order to construct $\prec$ based on the lengths of the agents'
strategies.

\begin{mechanism}[Length Game]\label{Length Game}
Form a linear order $\prec$, over the agents such
that if $|S_i|<|S_j|$, $i\prec
j$. If $|S_i|=|S_j|$ order the two in any order. Allocate agent $i$:
\begin{align}
A_i=S_i\backslash\bigcup_{j\prec i}S_j
\end{align}
\end{mechanism}

\pagebreak

\begin{proposition}
With sincere agents, Length Game produces a utilitarian optimal allocation.
\end{proposition}
\begin{proof}
Given situations $P$, let $A$ be the allocation produced by Length Game and $A'$
a different allocation. That is, there must be some interval $I\subseteq A_i$
such that $I\subseteq A'_j$ for $i\neq j$. Observe that $i\prec j$, otherwise
$I$ would have been in $A_j$. This means $|P_i|\leq|P_j|$. The efficiency gained
from $I$ in both cases is then:
\begin{align}
\frac{|I\cap P_i|}{|P_i|}\leq\frac{|I\cap P_j|}{|P_j|}
\end{align}
because $|I\cap P_i|=|I\cap P_j|=|I|$. As such $A'$ cannot improve on the
utilitarian efficiency of $A$.
\end{proof}

The more interesting question, of course, is what happens if the agents
strategise. That is, we want to find the equilibria of Length Game. To proceed
we need some more terminology. An $n$-tuple of strategies, $(S_1,...,S_n)$,
submitted to Length Game as input we shall call a \emph{profile}. The region of
the cake valued by agent $i$ and only agent $i$, $P_i\backslash\bigcup P_j$, is
agent $i$'s \emph{uncontested} region. Given a profile $(S_1,...,S_n)$ and the
resulting allocation $(A_1,...,A_n)$, if $S_i=A_i$ for all $i$, we say the
profile is \emph{reduced}.

We use the
standard notion of a pure strategy equilibrium: resistance to deviation by a
single agent.

\begin{definition}
Let $S=(S_1,...,S_n)$ be a profile and $A=(A_1,...,A_n)$ the resulting
allocation. We say that $S$ is in \emph{equilibrium} if there is no $i$ for
which there exists a $S'_i$ such that the profile $S'=(S_1,...,S'_i,...,S_n)$
produces an allocation $(A'_1,...,A'_n)$ where $u_i(A'_i)>u_i(A_i))$.

An allocation is an equilibrium if it is produced by an equilibrium profile.
\end{definition}

Reduced profiles are convenient because they simplify the strategic
considerations of the agents. If an agent sees a region of cake they could get
by claiming it, they should claim it, without paying heed to the other agents.

\begin{example}
Consider three agents, $P_1=[0,0.5]$, $P_2=[0,0.6]$, $P_3=[0.5,1]$. Consider the
non-reduced profile, $([0,0.4],[0.1,0.5],[0.5,1])$, where the tie is broken in
favour of agent 2. That is, the allocation is
$([0,0.1],[0.1,0.5],[0.5,1])$. It is not entirely clear what agent 2 should do.
They could claim some of $[0.5,0.6]$ and get it allocated to them instead of 3,
but by doing so they will lose their tie with 1 and the cake associated with it.

On the other hand, suppose the profile is $([0,0.1],[0.1,0.5],[0.5,1])$. Now
agent 2 has a clear incentive to claim $[0.5,0.6]$ as there can be no loss in
utility from doing so.
\end{example}

As such we would like to restrict our attention to reduced
profiles. We need two lemmata to show that there is no loss of generality in
doing so.

\pagebreak

\begin{lemma}
Given any profile of Length Game, there exists a reduced profile producing the
same allocation.
\end{lemma}
\begin{proof}
Let $S=(S_1,...,S_n)$ be a profile and $A=(A_1,...,A_n)$ the resulting
allocation. Construct $S'=(S'_1,...,S'_n)$ such that $S'_i=A_i$.

As portions are disjoint, so must be all $S'_i$. Thus regardless of the order
constructed by Length Game, given $S'$ the mechanism would simply allocate
$S'_i$ to agent $i$, as $|S'_i\backslash S'_j|=0$ for all $j$. Then $S'$ is a
reduced
profile producing $A$.
\end{proof}

\begin{lemma}
A Length Game profile is in equilibrium only if the associated reduced profile
is in equilibrium.
\end{lemma}
\begin{proof}
Let $S$ be a Length Game profile and $S'$ the associated reduced profile.
Observe that $|S_i|\geq |S'_i|$ for all $i$.

Suppose $S'$ is not an equilibrium profile. Then some agent $i$ has an available
strategy, $M_i$, such that $u_i(A^*_i)>u_i(A_i)$, where $A^*$ is the allocation
produced by $(S'_1,...,M_i,...,S'_n)$.

As $u_i(A^*_i)>u_i(A_i)$ there must be some sub-interval $I\subseteq A^*_i$,
$I\nsubseteq A_i$. There are two cases to consider.

If $I$ is not allocated to anyone in $S$, then clearly agent $i$ has incentive
to claim $I$ in $S$, so $S$ is not an equilibrium profile.

If $I\subseteq A_j$, it follows that $I\subseteq S'_j$, so if $i$ is allocated
$I$ by playing $M_i$ it must be the case that $|M_i|\leq|S'_j|$. Since
$|S'_j|\leq |S_j|$, $|M_i|\leq|S_j|$, and $i$ has incentive to play $M_i$ in
$S$, so it is not an equilibrium profile.
\end{proof}

With reduced profiles, we can clearly see that an agent has no incentive to
leave their uncontested region unclaimed.

\begin{lemma}\label{uncontested lemma}
A reduced Length Game profile with well behaved agents is in equilibrium only if
every agent's strategy includes their uncontested region.
\end{lemma}
\begin{proof}
Since agents are well behaved, no other agent would claim $i$'s uncontested
region, so $i$ can always increase their utility by claiming it.
\end{proof}

This gives us all the tools we need to characterise the equilibrium.

\begin{proposition}\label{LG equilibrium}
A well behaved reduced Length Game profile is in equilibrium if and only if:
\begin{enumerate}
\item
$\displaystyle\bigcup_{i=1}^n P_i\subseteq\bigcup_{i=1}^n S_i$ (All the valued
cake is allocated).
\item
Whenever $|S_i\cap P_j|\neq 0$ for $i\neq j$, $|S_i|\leq|S_j|$.
\end{enumerate}
\end{proposition}
\begin{proof}
For the if direction, assume 1 and 2 hold. Assume, for contradiction, that the
profile is not in equilibrium: some agent $i$ by submitting $S'_i\neq S_i$ can
force the allocation $A'$, where $u_i(A'_i)>u_i(A_i)$. Thus there is some
$I\subseteq A'_i, I\nsubseteq A_i$. $I$ cannot be in $i$'s uncontested region,
by Lemma \ref{uncontested lemma} we know that $I\subset S_i$, and since the
profile is reduced $S_i=A_i$. Therefore $I\subseteq P_j\cap P_i$, $j\neq i$. We
can choose this $j$ such that $|A_j\cap I|\neq 0$: by 1 we know that $I$ is
allocated, and by well behavedness we know that it is allocated to agents that
value it. Since $A_j=S_j$, we have that $|S_j\cap P_i|\neq0$. By 2, we have that
$|S_j|\leq|S_i|$.

Consider $|S'_i|$. We know that $A'_i\subseteq S'_i$, and $u_i(A'_i)>u_i(A_i)$,
so with \eqref{PU utility} we can derive that:
\begin{align}
\frac{|A'_i\cap P_i|}{|P_i|}>\frac{|A_i\cap P_i|}{|P_i|}
\end{align}
hence $|A'_i|>|A_i|$, and $|S'_i|>|S_i|\geq|S_j|$. But if $|S'_i|>|S_j|$, $j$
will come before $i$ in $\prec$, and be allocated $I$. This gives the desired
contradiction.

For the only if direction, first assume 1 fails. Then there is some cake valued
by some $i$ that is not allocated to any $j$. Agent $i$ can claim that cake and
increase their utility, so the profile is not in equilibrium. Next, assume 2
fails. Then there exist some $i,j$ such that $|S_i\cap P_j|\neq 0$ but
$|S_i|>|S_j|$. Agent $j$ can claim some of $|S_i\cap P_j|$ and be allocated it
because
$|S_i|>|S_j|$, so the profile is not in equilibrium.
\end{proof}

Note that with reduced profiles, we can use $S_i$ and $A_i$ interchangeably:
that is, Proposition \ref{LG equilibrium} also states the requirements for an
equilibrium allocation.

Barring pairwise disjoint preferences, equilibria are non-unique. In fact, there
are infinitely many of them. However, from a utilitarian point of view this is
largely irrelevant as all well behaved equilibria are identical in terms of
payoffs.

\begin{definition}
Two allocations, $A$ and $B$, are \emph{utilitarian equivalent} if
$u_i(A_i)=u_i(B_i)$ for all $i$.
\end{definition}

Before we prove that all well behaved equilibria are utilitarian equivalent we
need an auxiliary result, which is interesting in its own right.

\begin{proposition}\label{LGPareto}
All well behaved equilibria of Length Game are Pareto efficient.
\end{proposition}
\begin{proof}
Let us first observe that we cannot allocate more cake than is available.
Equivalently, we cannot allocate more valued cake than we have. Thus for any
allocation $A$:
\begin{align}
\sum_{i=1}^n|A_i\cap P_i|\leq\left|\ \bigcup_{i=1}^n P_i\ \right|
\end{align}

If $A$ is a Length Game equilibrium, we can strengthen the above into an
equality; we have seen in Proposition \ref{LG equilibrium} that all the valued
cake is allocated.

Now, suppose that $A'$ is a Pareto improvement on $A$. That is, for all $i$,
$u_i(A'_i)\geq u_i(A_i)$ and for some $j$, $u_j(A'_j)\geq u_j(A_j)$. Invoking
\eqref{PU utility}:
\begin{align}
\frac{|A'_i\cap P_i|}{|P_i|}\geq\frac{|A_i\cap P_i|}{|P_i|}
\end{align}
for all $i$, and:
\begin{align}
\frac{|A'_j\cap P_j|}{|P_j|}>\frac{|A_j\cap P_j|}{|P_j|}
\end{align}
for some $j$.

Multiplying out the denominators this is equivalent to $|A'_i\cap
P_i|\geq|A_i\cap P_i|$ and $|A'_j\cap P_j|>|A_j\cap P_j|$, and thus:
\begin{align}
\sum_{i=1}^n|A'_i\cap P_i|>\sum_{i=1}^n|A_i\cap P_i|=\left|\ \bigcup_{i=1}^n
P_i\ \right|
\end{align}
which gives us an impossibility.
\end{proof}

\begin{proposition}\label{LGequivalence}
All well behaved equilibria of Length Game are utilitarian equivalent.
\end{proposition}
\begin{proof}
Let $S$ and $S'$ be two reduced, well behaved equilibrium profiles that are not
utilitarian equivalent. Let $W$ be the set of all $i$ such that
$u_i(A'_i)>u_i(A_i)$. From Proposition \ref{LGPareto} we know that there is also
a non-empty set $L$ consisting of agents $j$ where $u_j(A'_j)<u_j(A_j)$. All
other agents
will constitute the set $N$.

Recall that in equilibria the entire valued cake is allocated, giving us the
identity:
\begin{align}
\sum_{i\in W}|A'_i|+\sum_{i\in N}|A'_i|+\sum_{i\in L}|A'_i|=\sum_{i\in
W}|A_i|+\sum_{i\in N}|A_i|+\sum_{i\in L}|A_i|
\end{align}
\begin{align}
\sum_{i\in W}|A'_i|+\sum_{i\in L}|A'_i|=\sum_{i\in W}|A_i|+\sum_{i\in L}|A_i|
\end{align}
\begin{align}\label{pos}
\sum_{i\in W}|A'_i|-\sum_{i\in W}|A_i|=\sum_{i\in L}|A_i|-\sum_{i\in L}|A'_i|
\end{align}

We claim that both sides of \eqref{pos} must be positive. If not, this would
mean that all agents in $L$ attain less utility in $A'$ despite the fact that
they have at least as much cake between them as they did before. This is only
possible if some agent receives some cake that they do not value, but this
cannot be the case because well behaved equilibria of Length Game are not
wasteful.

Consider a $j\in L$. As $u_j(A'_j)<u_j(A_j)$, there exists an interval
$I\subseteq P_j$ such that $I\subseteq A_j, I\nsubseteq A'_j$. However $I$ is
valued, so $I\subseteq A'_k$ for some $k$. As $S'$ is an equilibrium,
$|S'_k|\leq|S'_j|$. As $S$ is an equilibrium, $|S_j|\leq|S_k|$. Given that the
profiles are reduced and $u_j(A'_j)<u_j(A_j)$, we complete the chain to get
$|S'_k|\leq|S'_j|<|S_j|\leq|S_k|$. $|S'_k|<|S_k|$, hence $u_k(A'_k)<u_k(A_k)$,
hence $k\in L$. Thus the agents in $L$ do not lose any cake to agents outside of
it, and as allocations are not wasteful at least one agent in $L$ must be no
worse off in $A'$ than in $A$, giving us a contradiction.
\end{proof}

\section{Characterisations}\label{strchar}

Having shown that all well behaved equilibria of Length Game have the same
payoffs, we know wish to ask what these payoffs are. An important result is that
from a consequentialist perspective, this mechanism is not new. In equilibrium
it produces the same payoffs as Mechanism \ref{procaccia}.

\begin{proposition}
Consider a cake cutting situation with piecewise uniform preferences
$(P_1,...,P_n)$. Let $(A_1,...,A_n)$ be the allocation produced by Mechanism
\ref{procaccia}. Then $(A_1,...,A_n)$ is a well behaved equilibrium profile of
Length Game given the same preferences.
\end{proposition}
\begin{proof}
It suffices to show that $S=(A_1,...,A_n)$ satisfies the hypotheses of Proposition \ref{LG
equilibrium}. Recall that a subroutine of Mechanism \ref{procaccia} takes a
subset of agents, $\mathfrak{A}$, and of cake, $\mathfrak{X}$ as input.
$D(\mathfrak{A},\mathfrak{X})$ is all the regions of $\mathfrak{X}$ valued by at
least one agent in $\mathfrak{A}$ and:
\begin{align}
avg(\mathfrak{A},\mathfrak{X})=\frac{|D(\mathfrak{A},\mathfrak{X})|}{\#\mathfrak
{A } }
\end{align}

To see that all the valued cake is allocated, consider an arbitrary interval
$I\subseteq P_i$ for some $i$. Consider the subroutine on
$\mathfrak{A},\mathfrak{X}$ where $i\in\mathfrak{A}$. Let $I'\subseteq I$ be the
part of $I$ that was not yet allocated. That is, $I'=I\cap\mathfrak{X}$. As
$I'\subseteq P_i, I'\subseteq D(\mathfrak{A},\mathfrak{X})$. So after this
subroutine $I'$ will certainly be allocated to some agent.

To see that whenever $|A_i\cap P_j|\neq 0$ for $i\neq j$, $|A_i|\leq|A_j|$., we
consider two cases:

Case one: there is a subroutine on $\mathfrak{A}$ and $\mathfrak{X}$ such that
$i,j\in\mathfrak{A}$. The mechanism allocates all agents in $\mathfrak{A}$ a
portion of length $avg(\mathfrak{A},\mathfrak{X})$, so $|A_i|=|A_j|$.

Case two: there are subroutines on $\mathfrak{A}'$, $\mathfrak{X}'$ and
$\mathfrak{A}^*$, $\mathfrak{X}^*$ such that $i\in\mathfrak{A}'$ and
$j\in\mathfrak{A}^*$. Suppose $|A_i\cap P_j|\neq 0$. Then the subroutine on
$\mathfrak{A}'$, $\mathfrak{X}'$ must be executed first, because all of $P_j$
would be allocated after the subroutine on $\mathfrak{A}^*$. Since subroutines
are executed in order of increasing $avg(\mathfrak{A},\mathfrak{X})$,
$avg(\mathfrak{A}',\mathfrak{X}')\leq avg(\mathfrak{A}^*,\mathfrak{X}^*)$. So
$|A_i|\leq |A_j|$.
\end{proof}

\begin{corollary}
All well behaved equilibria of Length Game are envy free.
\end{corollary}
\begin{proof}
Mechanism \ref{procaccia} produces an envy free allocation which is also a well
behaved equilibrium. Proposition \ref{LGequivalence} establishes that all well
behaved equilibria of Length Game are utilitarian equivalent, so they must all
be envy free.
\end{proof}

\begin{corollary}
Mechanism \ref{procaccia} is Pareto efficient.
\end{corollary}
\begin{proof}
Proposition \ref{LGequivalence} and \ref{LGPareto}.
\end{proof}

This gives us two mechanisms which produce different allocations with sincere
agents, but have the same equilibria given strategic behaviour. One mechanism
induces truthful behaviour in agents, the other admits a dominant strategy
equilibrium. This should be reminiscent of a result from implementation theory:

\begin{theorem}[Revelation Principle]
Given a cake cutting situation with finitely representable preferences, if a
mechanism $M_1$ has dominant strategy equilibria, there exists a truthful
mechanism $M_2$ that produces utilitarian equivalent allocations.
\end{theorem}
\begin{proof}
This is a well known result in Economics. See, for instance, \cite{Myerson83}.

The idea behind the construction is that $M_2$, termed the
\emph{direct revelation} mechanism, takes the utility functions of the agents as
input. It then simulates the behaviour of the agents given $M_1$ and thus
produces the allocation.
\end{proof}

In our case Length Game is the na\"{i}ve mechanism, and Mechanism
\ref{procaccia} is the direct revelation equivalent. This is akin to the
contrast between first and second price auctions: with sincere agents the first
price auction generates higher revenue, while its equilibrium is equivalent to
that of the strategy-proof second price auction. Within auction theory, the
Revenue Equivalence theorem suggests that this is the best one can hope to
achieve. We prove a similar result with respect to the utilitarian efficiency of
piecewise uniform cake cutting situations.

\begin{theorem}
If a mechanism produces an allocation with a higher utilitarian
efficiency in some piecewise uniform situation than Mechanism \ref{procaccia},
there is a situation in which it produces an allocation with a lower utilitarian
efficiency.
\end{theorem}
\begin{proof}
The key in this proof is that if a mechanism produces allocation $(A_1,...,A_n)$
in situation $(P_1,...,P_n)$, then in any situation of the form
$(P_1,...,P'_i,...,P_n)$ agent $i$ can force allocation $(A_1,...,A_i,...,A_n)$
by casting the same strategy they would have cast had their preferences been
$P_i$.

Suppose that in situation $P=(P_1,...,P_n)$ Mechanism \ref{procaccia} produces
allocation $A=(A_1,...,A_n)$ while mechanism $M^*$ produces allocation
$A^*=(A^*_1,...,A^*_n)$ such that\\ $\mathcal{UE}(A^*)>\mathcal{UE}(A)$. There
must therefore be some $i$ such that $u_i(A_i^*)>u_i(A_i)$. Given Proposition
\ref{LGPareto}, there is also a $j$ with $u_j(A_j^*)<u_j(A_j)$. If there is more
than one such $j$, pick one with the largest $|P_j|$. We consider two cases.

Case one: $|P_j|<1$. Consider the situation $P'$ where $P'_i=[0,1]$,
$P'_k=P_k$ for $k\neq i$. Under mechanism \ref{procaccia} the allocation
produced is $A'$. Observe that $|A'_i|\leq|A_i|$: agent $i$ can force allocation
$A$, and since they do not it must be because they have no incentive in doing
so. Since $|A^*_i|>|A_i|$, $|A^*_i|>|A'_i|$. To achieve this, agent $i$ must be
given some cake by $M^*$ that Mechanism \ref{procaccia} gave to $j$ instead.
There is some interval $I\subseteq A'_j$, $I\subseteq A^*_i$. Observe that
giving $I$ to $i$ raises $i$'s utility by $|I|$, but lowers $j$'s by
$|I|/|P_j|$, $|P_j|< 1$. Since we picked $j$ to have the largest $|P_j|$ out of
all agents that lost utility, this means that
$\mathcal{UE}(A^*)<\mathcal{UE}(A')$.

Case two: $|P_j|=1$. Consider the situation $P'$ as above: agent $i$ with
preferences $[0,1]$ pretends to value $P_i$ to force allocation $A^*$ in $M^*$.
Suppose $j$ pretends to value $|P^*_j|<1$. If this will yield them a larger
slice, then $A^*$ is not an equilibrium. If this does not, then we fall back to
case one and in the situation $(P'_1,...,P^*_j,...P'_n)$, $M^*$ would attain a
lower utilitarian efficiency than Mechanism \ref{procaccia}.
\end{proof}

\chapter{Conclusion}\label{con}

We have summarised the main results in the history of cake cutting, placing our
focus on the mechanisms themselves, rather than measure theoretic existence
results. This has revealed that concerns of efficiency and truthfulness are
relatively new developments in the field.

As we have seen there may be good reasons for this: impossibility results abound
when optimal allocations are concerned. The general cake cutting model is too
broad to allow such allocations to be effected. Even if the desired optimum
exists, obtaining it may be impossible if it is against the interests of the
agents to do so.

We have, however, found a candidate for the ``next best" solution in the case of
piecewise uniform preferences: a mechanism that is never dominated by another,
first presented in \cite{Chen10} and given a superficially different, but
equivalent, characterisation here.

While the number of unresolved questions is vast, perhaps the most pertinent one
here is whether there are similar results for piecewise constant preferences.
Likewise, do non-trivial truthful mechanisms exist in such a case? The
restriction to piecewise uniform preferences allowed us to sidestep a plethora
of issues that would arise in such a situation. In fact, we conjecture that
given the simplicity of the preferences Mechanism \ref{procaccia} is in some
sense unique: perhaps all envy free, non-wasteful mechanisms must produce
allocations which coincide with the equilibria of Length Game. 

In a more general setting, even with the
combined tools
of Mathematics, Economics and Computer Science at our disposal it would seem
that further progress will be no cakewalk.

\appendix

\chapter{Table of Mechanisms}\label{table}

\begin{tabular}{l | l | l | l}
Mechanism & Comments & Introduced in & Page\\
\hline
Cut and Choose & 2 agents, R-W, Pr, EF, WT & Prehistoric
&\pageref{Cut and Choose}\\
Last Diminisher &  R-W, Pr, WT & \cite{Steinhaus48}
&\pageref{Last Diminisher}\\
Selfridge's Algorithm& 3 agents, R-W, Pr, EF,
WT&\cite{Woodall80}&\pageref{Selfridge's Algorithm}\\
Four Knives& 3 agents, MK, Pr, EF, WT&\cite{Stromquist80}&\pageref{Four
Swords}\\
Austin's Scheme& 2 agents, MK, Pr, EF, Eq,
WT&\cite{Austin82}&\pageref{Austin's Scheme}\\
Four Agent Moving Knife& 4 agents, MK, Pr, EF,
WT&\cite{Brams97}&\pageref{4 Agent Moving Knife}\\
Mechanism
\ref{Even nlogn}&R-W, Pr, WT&\cite{Even84}&\pageref{Even nlogn}\\
Mechanism
\ref{nonconstructive}&RP, Pr, EF, Eq,
Tr&\cite{Chen10},\cite{Mossel10}&\pageref{nonconstructive}\\
Mechanism
\ref{procaccia}&RP, Pr, EF,
Tr&\cite{Chen10}&\pageref{procaccia}\\
Lex Order&RP, Tr&Present work&\pageref{Lex Order}\\
Length Game&RP, Nv&Present work&\pageref{Length Game}\\
\end{tabular}
\vfill
\begin{tabular}{l l}
Legend &\\
\hline
R-W&Robertson-Webb protocol\\
MK&Moving knife protocol\\
RP&Revelation protocol\\
Pr&Proportional mechanism\\
EF&Envy free mechanism\\
Eq&Equitable mechanism\\
Nv&Na\"{i}ve mechanism\\
WT&Weakly truthful mechanism\\
Tr&Truthful mechanism\\
\end{tabular}

\chapter{Robertson-Webb Formulations}\label{RW}

\setlength{\parindent}{0in}

In this appendix we give Robertson-Webb formulations of mechanisms that appear
in this text. Recall that the allowed queries are \texttt{eval(a,b)} and
\texttt{cut(a,x)}. We will use subscripts to indicate the agent queried. That
is, \texttt{eval$_\textrm{\texttt{i}}$(a,b)} would query agent $i$ to evaluate
the slice $[a,b]$. If it is understood that $X=[x_1,x_2]$ is a slice, we may
write \texttt{eval$_\textrm{\texttt{i}}$(X)} instead of
\texttt{eval$_\textrm{\texttt{i}}$(x$_1$,x$_2$)}. As before, $A$ is the
allocation and $A_i$ is the portion of
agent $i$ in $A$. $\mathcal{A}$ is the set of agents. $\#\mathcal{A}=n$.\\

{\bf Mechanism \ref{Cut and Choose}: Cut and Choose}\\
\vspace{-0.3cm}\\
$a=$\texttt{cut$_\textrm{\texttt{1}}$(0,0.5)}\\
{\bf
if }\texttt{eval$_\textrm{\texttt{2}}$(0,a)}$>$\texttt{eval$_\textrm{\texttt{2}}
$(a,1)}:\\
\hspace*{0.5cm}$A_1=[a,1]$, $A_2=[0,a]$\\
{\bf else: }\\
\hspace*{0.5cm}$A_1=[0,a]$, $A_2=[a,1]$\\
{\bf return }$A$\\

{\bf Mechanism \ref{Last Diminisher}: Last Diminisher}\\
\vspace{-0.3cm}\\
$s=0$\\
$l=1$\\
{\bf while }$\mathcal{A}\neq\emptyset$\\
\hspace*{0.5cm}{\bf for }$i$ in $\mathcal{A}:$\\
\hspace*{1cm}{\bf if }\texttt{eval$_\textrm{\texttt{i}}$(s,l)}$>1/n$:\\
\hspace*{1.5cm}last=$i$\\
\hspace*{1.5cm}{\bf if }$\#\mathcal{A}>1$:\\
\hspace*{2cm}$l=$\texttt{cut$_\textrm{\texttt{1}}$(s,1/n)}\\
\hspace*{0.5cm}$A_{\textrm{last}}=[s,l]$\\
\hspace*{0.5cm}$s=l$\\
\hspace*{0.5cm}$l=1$\\
\hspace*{0.5cm}$\mathcal{A}=\mathcal{A}\backslash \{$last$\}$\\
{\bf return }$A$\\

\pagebreak

{\bf Mechanism \ref{Selfridge's Algorithm}: Selfridge's Algorithm}\\
\vspace{-0.3cm}\\
$a=$\texttt{cut$_\textrm{\texttt{1}}$(0,1/3)}\\
$b=$\texttt{cut$_\textrm{\texttt{1}}$(a,1/3)}\\
$X=$arg\hspace*{-0.7cm}$\displaystyle\max_{\{[0,a],[a,b],[b,1]\}}$\texttt{
eval$_\textrm{\texttt{2}}$(x)}$=[x_1,x_2]$\\
$Z=$arg\hspace*{-0.7cm}$\displaystyle\min_{\{[0,a],[a,b],[b,1]\}}$\texttt{
eval$_\textrm{\texttt{2}}$(x)}$=[z_1,z_2]$\\
$Y=[0,1]\backslash X\cup Z=[y_1,y_2]$\\
$v=$\texttt{eval$_\textrm{\texttt{2}}$(Y)}\\
$c=$\texttt{cut$_\textrm{\texttt{2}}$(x$_1$,Y)}\\
$X'=[x_1,c]$\\
{\bf if
}\texttt{eval$_\textrm{\texttt{3}}$(X')}$\geq$\texttt{eval$_\textrm{\texttt{3}}
$(Y)} and
\texttt{eval$_\textrm{\texttt{3}}$(X')}$\geq$\texttt{eval$_\textrm{\texttt{3}}
$(Y)}:\\
\hspace*{0.5cm}$A_3=A_3\cup{X'}$\\
\hspace*{0.5cm}Case$=1$\\
\hspace*{0.5cm}{\bf if
}\texttt{eval$_\textrm{\texttt{1}}$(Y)}$\geq$\texttt{eval$_\textrm{\texttt{1}}
$(Z)}\\
\hspace*{1cm}$A_1=A_1\cup{Y}$\\
\hspace*{1cm}$A_2=A_2\cup{Z}$\\
\hspace*{0.5cm}{\bf else }:\\
\hspace*{1cm}$A_1=A_1\cup{Z}$\\
\hspace*{1cm}$A_2=A_2\cup{Y}$\\
{\bf else if
}\texttt{eval$_\textrm{\texttt{3}}$(Y)}$\geq$\texttt{eval$_\textrm{\texttt{3}}
$(X')} and
\texttt{eval$_\textrm{\texttt{3}}$(Y)}$\geq$\texttt{eval$_\textrm{\texttt{3}}
$(Z)}:\\
\hspace*{0.5cm}$A_3=A_3\cup{Y}$\\
\hspace*{0.5cm}{\bf if
}\texttt{eval$_\textrm{\texttt{1}}$(Z)}$\geq$\texttt{eval$_\textrm{\texttt{1}}
$(X')}\\
\hspace*{1cm}$A_1=A_1\cup{Z}$\\
\hspace*{1cm}$A_2=A_2\cup{X'}$\\
\hspace*{0.5cm}Case$=2$\\
\hspace*{0.5cm}{\bf else }:\\
\hspace*{1cm}$A_1=A_1\cup{X'}$\\
\hspace*{1cm}$A_2=A_2\cup{Z}$\\
\hspace*{1cm}Cut and Choose($\{2,3\},[c,x_2]$)\\
{\bf else}:\\
\hspace*{0.5cm}$A_3=A_3\cup{Z}$\\
\hspace*{0.5cm}{\bf if
}\texttt{eval$_\textrm{\texttt{1}}$(Y)}$\geq$\texttt{eval$_\textrm{\texttt{1}}
$(X')}\\
\hspace*{1cm}$A_1=A_1\cup{Y}$\\
\hspace*{1cm}$A_2=A_2\cup{X'}$\\
\hspace*{0.5cm}Case$=2$\\
\hspace*{0.5cm}{\bf else }:\\
\hspace*{1cm}$A_1=A_1\cup{X'}$\\
\hspace*{1cm}$A_2=A_2\cup{Y}$\\
\hspace*{1cm}Cut and Choose($\{2,3\},[c,x_2]$)\\
{\bf if }Case$==1$:\\
\hspace*{0.5cm}$u=$\texttt{eval$_\textrm{\texttt{2}}$([c,x$_2$])}\\
\hspace*{0.5cm}$d=$\texttt{cut$_\textrm{\texttt{2}}$(0,1/3*u)}\\
\hspace*{0.5cm}$e=$\texttt{cut$_\textrm{\texttt{2}}$(d,1/3*u)}\\
\hspace*{0.5cm}{\bf if
}\texttt{eval$_\textrm{\texttt{3}}$(c,d)}$\geq$\texttt{eval$_\textrm{\texttt{3}}
$(d,e)} and
\texttt{eval$_\textrm{\texttt{3}}$(c,d)}$\geq$\texttt{eval$_\textrm{\texttt{3}}
$(e,x$_2$)}:\\
\hspace*{1cm}$A_3=A_3\cup{[c,d]}$\\
\hspace*{1cm}{\bf if
}\texttt{eval$_\textrm{\texttt{1}}$(d,e)}$\geq$\texttt{eval$_\textrm{\texttt{1}}
$(e,x$_2$)}\\
\hspace*{1.5cm}$A_1=A_1\cup{[d,e]}$\\
\hspace*{1.5cm}$A_2=A_2\cup{[e,x_2]}$\\
\hspace*{1cm}{\bf else }:\\
\hspace*{1.5cm}$A_1=A_1\cup{[e,x_2]}$\\
\hspace*{1.5cm}$A_2=A_2\cup{[d,e]}$\\
\hspace*{0.5cm}{\bf else if
}\texttt{eval$_\textrm{\texttt{3}}$(d,e)}$\geq$\texttt{eval$_\textrm{\texttt{3}}
$(c,d)} and
\texttt{eval$_\textrm{\texttt{3}}$(d,e)}$\geq$\texttt{eval$_\textrm{\texttt{3}}
$(e,x$_2$)}:\\
\hspace*{1cm}$A_3=A_3\cup{[d,e]}$\\
\hspace*{1cm}{\bf if
}\texttt{eval$_\textrm{\texttt{1}}$(c,d)}$\geq$\texttt{eval$_\textrm{\texttt{1}}
$(e,x$_2$)}\\
\hspace*{1.5cm}$A_1=A_1\cup{[c,d]}$\\
\hspace*{1.5cm}$A_2=A_2\cup{[e,x_2]}$\\
\hspace*{1cm}{\bf else }:\\
\hspace*{1.5cm}$A_1=A_1\cup{[e,x_2]}$\\
\hspace*{1.5cm}$A_2=A_2\cup{[c,d]}$\\
\hspace*{0.5cm}{\bf else }:\\
\hspace*{1cm}$A_3=A_3\cup{[e,x_2]}$\\
\hspace*{1cm}{\bf if
}\texttt{eval$_\textrm{\texttt{1}}$(d,e)}$\geq$\texttt{eval$_\textrm{\texttt{1}}
$(c,d)}\\
\hspace*{1.5cm}$A_1=A_1\cup{[d,e]}$\\
\hspace*{1.5cm}$A_2=A_2\cup{[c,d]}$\\
\hspace*{1cm}{\bf else }:\\
\hspace*{1.5cm}$A_1=A_1\cup{[c,d]}$\\
\hspace*{1.5cm}$A_2=A_2\cup{[d,e]}$\\
{\bf else }:\\
\hspace*{0.5cm}$u=$\texttt{eval$_\textrm{\texttt{3}}$([c,x$_2$])}\\
\hspace*{0.5cm}$d=$\texttt{cut$_\textrm{\texttt{3}}$(0,1/3*u)}\\
\hspace*{0.5cm}$e=$\texttt{cut$_\textrm{\texttt{3}}$(d,1/3*u)}\\
\hspace*{0.5cm}{\bf if
}\texttt{eval$_\textrm{\texttt{2}}$(c,d)}$\geq$\texttt{eval$_\textrm{\texttt{2}}
$(d,e)} and
\texttt{eval$_\textrm{\texttt{2}}$(c,d)}$\geq$\texttt{eval$_\textrm{\texttt{2}}
$(e,x$_2$)}:\\
\hspace*{1cm}$A_2=A_2\cup{[c,d]}$\\
\hspace*{1cm}{\bf if
}\texttt{eval$_\textrm{\texttt{1}}$(d,e)}$\geq$\texttt{eval$_\textrm{\texttt{1}}
$(e,x$_2$)}\\
\hspace*{1.5cm}$A_1=A_1\cup{[d,e]}$\\
\hspace*{1.5cm}$A_3=A_3\cup{[e,x_2]}$\\
\hspace*{1cm}{\bf else }:\\
\hspace*{1.5cm}$A_1=A_1\cup{[e,x_2]}$\\
\hspace*{1.5cm}$A_3=A_3\cup{[d,e]}$\\
\hspace*{0.5cm}{\bf else if
}\texttt{eval$_\textrm{\texttt{2}}$(d,e)}$\geq$\texttt{eval$_\textrm{\texttt{2}}
$(c,d)} and
\texttt{eval$_\textrm{\texttt{2}}$(d,e)}$\geq$\texttt{eval$_\textrm{\texttt{2}}
$(e,x$_2$)}:\\
\hspace*{1cm}$A_2=A_2\cup{[d,e]}$\\
\hspace*{1cm}{\bf if
}\texttt{eval$_\textrm{\texttt{1}}$(c,d)}$\geq$\texttt{eval$_\textrm{\texttt{1}}
$(e,x$_2$)}\\
\hspace*{1.5cm}$A_1=A_1\cup{[c,d]}$\\
\hspace*{1.5cm}$A_3=A_3\cup{[e,x_2]}$\\
\hspace*{1cm}{\bf else }:\\
\hspace*{1.5cm}$A_1=A_1\cup{[e,x_2]}$\\
\hspace*{1.5cm}$A_3=A_3\cup{[c,d]}$\\
\hspace*{0.5cm}{\bf else }:\\
\hspace*{1cm}$A_2=A_2\cup{[e,x_2]}$\\
\hspace*{1cm}{\bf if
}\texttt{eval$_\textrm{\texttt{1}}$(d,e)}$\geq$\texttt{eval$_\textrm{\texttt{1}}
$(c,d)}\\
\hspace*{1.5cm}$A_1=A_1\cup{[d,e]}$\\
\hspace*{1.5cm}$A_3=A_3\cup{[c,d]}$\\
\hspace*{1cm}{\bf else }:\\
\hspace*{1.5cm}$A_1=A_1\cup{[c,d]}$\\
\hspace*{1.5cm}$A_3=A_3\cup{[d,e]}$\\
{\bf return }$A$\\

{\bf Mechanism \ref{Even nlogn}}\\
\vspace{-0.3cm}\\
Subroutine($\mathcal{A}$,$[0,1]$)\\
{\bf return }$A$\\
\vspace{-0.15cm}\\
{\bf Subroutine($\mathfrak{A},[s,t]$):}\\
{\bf if }$\#\mathfrak{A}=1$\\
\hspace*{0.5cm}$A_i=[s,t]$, $i\in\mathfrak{A}$\\
\hspace*{0.5cm}{\bf break}\\
{\bf for }$i$ in $\mathfrak{A}:$\\
\hspace*{0.5cm}$v=$\texttt{eval$_\textrm{\texttt{i}}$(s,t)}\\
\hspace*{0.5cm}$m_i=$\texttt{cut$_\textrm{\texttt{i}}$(s,1/2*v)}\\
Form array $L$ where $L[i]\in\mathfrak{A}$ and $m_{L[i]}<m_{L[j]}$ implies
$i<j$.\\
{\bf for }$j\leq\#\mathfrak{A}/2:$\\
\hspace*{0.5cm}$\mathfrak{A}=\mathfrak{A}\cup\{L[j]\}$\\
Subroutine($\mathfrak{A_1}$,$[s,m_{\lfloor\#\mathfrak{A}/2\rfloor}]$)\\
Subroutine($\mathfrak{A}\backslash\mathfrak{A_1}$,$[m_{\lfloor\#\mathfrak{A}
/2\rfloor},t]$)\\

\bibliographystyle{alpha}
\bibliography{references.bib}

\end{document}